\documentclass[11pt]{article}
\usepackage[all]{xy}
\usepackage{amsfonts}
\usepackage{amsthm}
\usepackage{enumerate}
\usepackage{graphicx}
\usepackage{mathrsfs}
\usepackage{bm}
\usepackage{cite}
\usepackage{amssymb,amsmath} % font used for R in Real numbers
\usepackage{makecell}
\usepackage{tikz}
\usepackage{pgf}
\usepackage{tikz}
\usetikzlibrary{patterns}
\usepackage{pgffor}
\usepackage{pgfcalendar}
\usepackage{pgfpages}
\usepackage{shuffle,yfonts}
\usepackage{mathtools}
\DeclareFontFamily{U}{shuffle}{}
\DeclareFontShape{U}{shuffle}{m}{n}{ <-8>shuffle7 <8->shuffle10}{}

\newcommand{\bfs}{{\boldsymbol{\sl{s}}}}

 \allowdisplaybreaks
\usetikzlibrary{arrows,shapes,chains}
 \allowdisplaybreaks

\catcode`!=11
\let\!int\int \def\int{\displaystyle\!int}
\let\!lim\lim \def\lim{\displaystyle\!lim}
\let\!sum\sum \def\sum{\displaystyle\!sum}
\let\!sup\sup \def\sup{\displaystyle\!sup}
\let\!inf\inf \def\inf{\displaystyle\!inf}
\let\!cap\cap \def\cap{\displaystyle\!cap}
\let\!max\max \def\max{\displaystyle\!max}
\let\!min\min \def\min{\displaystyle\!min}
\let\!frac\frac \def\frac{\displaystyle\!frac}
\catcode`!=12

\let\oldsection\section
\renewcommand\section{\setcounter{equation}{0}\oldsection}

\allowdisplaybreaks

\def\R{\mathbb{R}}

\def\N{\mathbb{N}}

\def\Q{\mathbb{Q}}

\def\ze{\zeta}

\theoremstyle{plain}
\newtheorem{thm}{Theorem}[section]
\newtheorem{lem}[thm]{Lemma}
\newtheorem{cor}[thm]{Corollary}

\newtheorem{pro}[thm]{Proposition}
\theoremstyle{definition}
\newtheorem{defn}{Definition}[section]
\newtheorem{re}[thm]{Remark}
\newtheorem{exa}[thm]{Example}

\begin{document}
	%%%%%%%%%%%%%%%%%%%% title %%%%%%%%%%%%%%%%%%%%%%%%%%%%%%%%%%%%%%%%%%%%%%%%
	\title{\bf A Family of Berndt-Type Integrals and Associated Barnes Multiple Zeta Functions}
	\author{ {Xinyue Gu${}^{a,}$\thanks{Email: gxy2024@ahnu.edu.cn}},\quad {Ce Xu${}^{a,}$\thanks{Email: cexu2020@ahnu.edu.cn}}\quad
		and { Jianing Zhou${}^{b}$\thanks{Email: 202421511272@smail.xtu.edu.cn}}\\[1mm]
		\small a. School of Mathematics and Statistics, Anhui Normal University, \\ \small Wuhu 241002, P.R. China\\
		\small b. School of
		Mathematics and Computational Science, Xiangtan University, \\ \small Xiangtan 411105,  P.R. China}
	
	\date{}
	\maketitle
	
	\noindent{\bf Abstract.} In this paper, we focus on the calculation of a specific type of Berndt integral, which exclusively involves (hyperbolic) cosine functions. Initially, this integral is transformed into a Ramanujan-type hyperbolic (infinite) sum via contour integration. Subsequently, a function incorporating \(\theta\) is defined. By employing the residue theorem, the mixed Ramanujan-type hyperbolic (infinite) sum with both hyperbolic cosine and hyperbolic sine in the denominator is converted into a simpler Ramanujan-type hyperbolic (infinite) sum, which contains only hyperbolic cosine or hyperbolic sine in the denominator. The simpler Ramanujan-type hyperbolic (infinite) sum is then evaluated using Jacobi elliptic functions, Fourier series expansions, and Maclaurin series expansions. Ultimately, the result is expressed as a rational polynomial of \(\Gamma(1/4)\) and \(\pi^{-1/2}\). Additionally, the integral is related to the Barnes multiple zeta function, which provides an alternative method for its calculation.\\
	\medskip
	\noindent{\bf Keywords}: Berndt-type integral, Hyperbolic functions, Contour integration, Jacobi elliptic functions, Barnes multiple zeta functions.

	\noindent{\bf AMS Subject Classifications (2020):}  33E05, 33E20, 44A05, 11M99.

	\section{Introduction}
In \cite{XZ2024}, the second author of this paper and Zhao referred to integrals of the following form as $m$-th order Berndt-type integrals, since it was approximately a decade ago that Professor Berndt first systematically studied the case when $m=1$. For positive integers $s$ and $m$, the \emph{Berndt-type integrals} of order $m$ is defined by
\begin{align}\label{BTI-definition-1}
{\rm BI}^{\pm}(s,m):=\int_0^\infty \frac{x^{s-1}\mathrm{d}x}{(\cos x\pm\cosh x)^m},
\end{align}
where $s\geq 1$ and $m\geq 1$ for the ``$+$" case and $s\geq 2m+1$ otherwise. In fact, the study of Berndt-type integrals can be traced back to 1912, when the renowned Indian mathematician Ramanujan investigated the following improper integrals involving trigonometric and hyperbolic functions. He conjectured the integral identity (see \cite[pp. 325-326]{Rama1916})
\begin{equation*}%\label{inte-Ramanujan}
		\int_0^\infty \frac{\sin(nx)}{x(\cos x+\cosh x)}\mathrm{d}x=\frac{\pi}{4}\quad (n\in\N),
	\end{equation*}
which was only rigorously proven four years later by Wilkinson \cite{W1916}.

About 80 years afterward, the study of such integrals gained renewed interest with the work of Ismail and Valent \cite{Ismail1998}, who discovered another remarkable integral (\emph{Ismail-type integrals}), exemplified by
	\begin{align}\label{A.KuznetsovInt}
		\int_{-\infty}^{\infty}{\frac{\mathrm{d}t}{\cos \left( K\sqrt{t} \right) +\cosh \left( K'\sqrt{t} \right)}=2}.
	\end{align}
Here $K\equiv K(x)$ denotes the \emph{complete elliptic integral of the first kind} and $K':=K(1-x)$. Kuznetsov \cite{K2017} generalized Ismail-type integrals using contour integration and theta functions, proving
	\begin{align}\label{Ims-Kuz+}
		\frac1{2}\int_{-\infty}^\infty \frac{t^n \mathrm{d}t}{\cos(K\sqrt{t})+\cosh(K'\sqrt{t})}=(-1)^n \frac{\mathrm{d}^{2n+1}}{\mathrm{d}u^{2n+1}} \frac{{\rm sn} (u,k)}{{\rm cd}(u,k)}|_{u=0},
	\end{align}
	where $x=k^2\ (0<k<1)$ and ${\rm sn} (u,k)$ is the \emph{Jacobi elliptic function} defined through the inversion of the elliptic integral
	\begin{align}
		u=\int_0^\varphi \frac{\mathrm{d}t}{\sqrt{1-k^2\sin^2 t}}\quad (0<k^2<1),
	\end{align}
	that is, ${\rm sn}(u):=\sin \varphi$. As before, $k$ is referred to as the elliptic modulus. We also write $\varphi={\rm am}(u,k)={\rm am}(u)$ and call it the Jacobi amplitude.
	The Jacobi elliptic function ${\rm cd} (u,k)$ can be defined as follows
	\begin{align*}
		& {\rm cd}(u,k):=\frac{{\rm cn}(u,k)}{{\rm dn}(u,k)},
	\end{align*}
	where ${\rm cn}(u,k):=\sqrt{1-{\rm sn}^2 (u,k)}$ and ${\rm dn} (u,k):=\sqrt{1-k^2{\rm sn}^2(u,k)}$.

In the past decade, research on these analogous integrals has  attracted increasing interest, yielding notable progress. Recent advances in this direction can be found in the works of \cite{BV2024,PanWang2025,RXYZ2024,RXZ2023,XZ2023,XZ2024,ZhangRui2024,Zhou2025} and the references therein. A significant turning point in this line of research occurred in 2016 when Professor Berndt \cite{Berndt2016} developed a systematic approach to study ${\rm BI}^{\pm}(s,1)$. Berndt further advanced the field by evaluating such integrals directly using Cauchy's residue theory, Fourier series expansions, and Maclaurin series expansions of Jacobi elliptic functions. In their series of works \cite{RXZ2023,XZ2023,XZ2024}, Rui, Xu, and Zhao developed an extended framework based on Berndt's approach, which led to fundamental structural theorems for higher-order Berndt-type integrals. Letting $X=\Gamma^4(1/4)$ and $Y=\pi^{-1}$, they proved that for all integers $m\geq 1$ and $p\geq [m/2]$ (see \cite[Thms 1.1 and 1.2]{XZ2024}),
\begin{align*}
{\rm BI}^+(4p+2,m)\in\mathbb{Q}[X,Y]\quad\text{and}\quad {\rm BI}^-(a,m)\in\mathbb{Q}[X,Y] \quad (0<a-2m\equiv 1 \pmod{4}),
\end{align*}
and the optimal range for the degrees of the polynomials in $X$ and $Y$ has been precisely determined.	

Very recently, a groundbreaking connection between Berndt-type integrals and \emph{Barnes multiple zeta functions} (See Section \ref{sec-Berndt-Barmzf}) was established by Bradshaw and Vignat \cite[Cor. 5]{BV2024}, who also linked Jacobi elliptic functions to Ismail-type integrals: for $n\in \N_0:=\N\cup \{0\}$,

	\begin{align}\label{Ims-Kuz-1}
		\frac1{2}\int_{-\infty}^\infty \frac{x^{n+1} \mathrm{d}x}{\cos(\sqrt{x}K)-\cosh(\sqrt{x}K')}=(-1)^{n+1}4 \frac{\mathrm{d}^{2n+1}}{\mathrm{d}u^{2n+1}} \frac{{\rm sn}^2 (u,k)}{{\rm cd}^2(u,k){\rm sd}(2u,k)}|_{u=0},
	\end{align}
	where the Jacobi elliptic function ${\rm sd}(u)\equiv {\rm sd}(u,k)$ is defined by
	\begin{align}\label{defin-sd}
		{\rm sd}(u,k):=\frac{{\rm sn}(u,k)}{{\rm dn}(u,k)}.
	\end{align}
	
As shown in the preceding analysis, the integrands of the previously studied Berndt-type integrals consist solely of additive or subtractive combinations of a single trigonometric cosine and a single hyperbolic cosine function

	In this paper, we focus on integrals of the form:
	\begin{align}\label{defn-Berndt-Int}
		\int_0^{\infty}{\frac{x^{s}\mathrm{d}x}{\left[ \cosh(2x)-\cos(2x) \right] \left[ \cosh x+\cos x \right]}},
	\end{align}
which feature purely (hyperbolic) cosine integrands and the integrand involves the product of two terms: (i) the sum of a cosine and a hyperbolic cosine function, and (ii) their difference.
By leveraging contour integration and series expansions of Jacobi elliptic functions, we derive structural theorems and explicit evaluations for these Berndt-type integrals. We are going to prove the following evaluations (see Theorem \ref{mixInt1}): let $\Gamma:=\Gamma(1/4)$, for $m\in \N$,
	\begin{align}\label{StrR-Berndt-Int}
		\int_0^{\infty}{\frac{x^{4m-1}\mathrm{d}x}{\left[ \cosh(2x)-\cos(2x) \right] \left[ \cosh x+\cos x \right]}}&\in \mathbb{Q} \frac{\Gamma ^{8m-4}}{\pi ^{2m-1}}+ \frac{\mathbb{Q}}{\sqrt{2}}\frac{\Gamma ^{8m-2}}{\pi^{(4m-1)/2}}\nonumber\\
		&\quad+\mathbb{Q}\frac{\Gamma ^{8m}}{\pi^{2m}}+ \frac{\mathbb{Q}}{\sqrt{2}}\frac{\Gamma ^{8m+2}}{\pi^{(4m+3)/2}}+\mathbb{Q} \frac{\Gamma ^{8m+4}}{\pi^{2m+3}}.
	\end{align}
Furthermore, we compute the specific coefficients preceding each term and set up some explicit relations between the Berndt-type integrals and Barnes multiple zeta functions. Additionally, we establish their relationship with Barnes multiple zeta functions, offering new closed-form evaluations and insights (see Theorem \ref{GBZeta}). Furthermore, it should be noted that the last two authors of this paper, together with Chen in a recent joint work \cite{CXZ2026}, applied methods similar to those used in the present article to obtain a structural result analogous to \eqref{StrR-Berndt-Int} for a mixed Berndt-type integral. Specifically, in [1], the denominator in \eqref{defn-Berndt-Int} was replaced by
$
\bigl[ \cosh(2x) - \cos(2x) \bigr] \bigl[ \cosh x - \cos x \bigr],
$
and the numerator exponent was taken as \( s = 4m - 3 \).

	\section{Some Preliminary Results}
	In this section, we present fundamental mathematical tools and special functions that form the basis for our subsequent analysis. Let $_2F_1\left( a,b;c;x \right)$ denote the \emph{Gaussian hypergeometric function} defined by
	\begin{equation*}
		{_2}F_1(a,b;c;x)=\sum\limits_{n=0}^\infty \frac{(a)_n(b)_n}{(c)_n} \frac {x^n}{n!}\quad (a,b,c\in\mathbb{C}),
	\end{equation*}
	where $\left( a \right) _0:=1$ and $\left( a \right)_n=\frac{\Gamma \left( a+n \right)}{\Gamma \left( a \right)}=\prod_{j=0}^{n-1}{\left( a+j \right)}\quad (n\in \mathbb{N}),
	$ and $ \Gamma(x)$ is the \emph{Gamma function} defined by
	$$\Gamma \left( z \right) :=\int_0^{\infty}{e^{-t} t^{z-1}\mathrm{d}t}\quad \left( \Re \left( z \right) >0 \right).$$
	We know that $_2F_1\left( a,b;c;x \right)$  plays a crucial role in evaluating elliptic integrals, particularly the  \emph{complete elliptic integral of the first and second kinds} are defined by (Whittaker and Watson \cite{WW1966})
	\begin{equation*}
	K:=K\left( k^2 \right) =\int_0^{\pi /2}{\frac{\mathrm{d}\varphi}{\sqrt{1-k^2\sin ^2\varphi}}=\frac{\pi}{2}}\,\,_2F_1\left(
	\frac{1}{2},\frac{1}{2};1;k^2 \right) ,
	\end{equation*}

    \begin{equation*}
		E:=E(k^2):=\int_{0}^{\pi/2} \sqrt{1-k^2\sin^2\varphi}\mathrm{d}\varphi=\frac {\pi}{2} {_2}F_{1}\left(-\frac {1}{2},\frac {1}{2};1;k^2\right).
	\end{equation*}
	
Following Ramanujan's notation, we define key variables and their relationships:
	\begin{align}\label{notations-Ramanujan}
		x:=k^2,\ y\left( x \right):=\pi K'/K,\ q\equiv q\left( x \right):=e^{-y},\ z:=z\left( x \right) =2K/\pi,\ z'=\mathrm{d}z/\mathrm{d}x,
	\end{align}
    where $ K'=K(1-x)$.
	Then$$
	\frac{\mathrm{d}^nz}{\mathrm{d}x^n}=\frac{\left( 1/2 \right) _{n}^{2}}{n!} \,_{2}F_1\left( \frac{1}{2}+n,\frac{1}{2}+n,1+n;x \right) .
	$$
	Applying the identity (\cite{A2000})
	$$
	\,_{2}F_1\left( a,b,\frac{a+b+1}{2};\frac{1}{2} \right) =\frac{\Gamma \left( {1}/{2} \right) \Gamma \left( {(a+b+1)}/{2} \right)}{\Gamma \left( {(a+1)}/{2} \right) \Gamma \left( {(b+1)}/{2} \right)},
	$$
	these satisfy important differential properties
	\begin{align}\label{x=1/2,y=pi}
		\frac{\mathrm{d}^nz}{\mathrm{d}x^n}\mid_{x=\tfrac{1}{2}}^{}=\frac{\left( 1/2 \right) _{n}^{2}\sqrt{\pi}}{\Gamma ^2\left( \frac{n}{2}+\frac{3}{4} \right)}.
	\end{align}
	
	In particular, taking $n = 0,1,2$ in (\ref{x=1/2,y=pi}), we obtain that
	$$
	y=\pi, \quad z\left(\frac{1}{2}\right)=\frac{\Gamma^{2}(1 / 4)}{2 \pi^{3 / 2}}, \quad z^{\prime}\left(\frac{1}{2}\right)=\frac{4 \sqrt{\pi}}{\Gamma^{2}(1 / 4)}, \quad z^{\prime \prime}\left(\frac{1}{2}\right)=\frac{\Gamma^{2}(1 / 4)}{2 \pi^{3 / 2}},
	$$
	which are essential for our explicit evaluations. The derivative relation (see \cite[pp. 120, Entry 9(i)]{B1991}):
    \begin{equation}\label{dx/dy}
    \frac{\mathrm{d}x}{\mathrm{d}y}=-x(1-x)z^{2}
    \end{equation}
	provides a crucial connection between these variables that we exploit in later calculations.

	\section{Some Definitions and Lemmas}
	This section establishes the core mathematical framework for our subsequent analysis through carefully constructed definitions and technical lemmas.
    \begin{defn}
		Let $s\in\mathbb{C}$, $g_1,g_2,g_3,g_4\in\mathbb{N}_0$ and $a,b,\theta\in\mathbb{R}$ with $|\theta| <2b\pi$ and $a,b\neq 0$. Define
		\begin{align*}
			&Z(s,\theta;a,b|g_1,g_2,g_3,g_4):=\frac{\pi^{g_1+g_2+g_3+g_4}\sinh(\theta s)}{\sin^{g_1}(a\pi s)\cos^{g_2}(a\pi s)\sinh^{g_3}(b\pi s)\cosh^{g_4}(b\pi s)},\\&\bar{Z}(s,\theta;a,b|g_1,g_2,g_3,g_4):=\frac{\pi^{g_1+g_2+g_3+g_4}\cosh(\theta s)}{\sin^{g_1}(a\pi s)\cos^{g_2}(a\pi s)\sinh^{g_3}(b\pi s)\cosh^{g_4}(b\pi s)}.
		\end{align*}
		
	\end{defn}
	\begin{defn}
		Let $m_1,m_2\in\mathbb{Z}$, $m_1+m_2\geq1$ and $p\in\mathbb{Z}$. Define
		\begin{align*}
			G_{p,m_1,m_2}(y)&:=\sum_{n = 1}^{\infty}\frac{(-1)^n n^p}{\sinh^{m_1}(ny)\cosh^{m_2}(ny)},\\
			\tilde{G}_{p,m_1,m_2}(y)&:=\sum_{n = 1}^{\infty}\frac{(-1)^n (2n-1)^{p}}{\sinh^{m_1}((2n-1)y/2)\cosh^{m_2}((2n-1)y/2)}.
		\end{align*}
	\end{defn}
	
	\begin{lem}\emph{(\cite{PFBE1998})}\label{Sum Res=0}
		Let $\xi(s)$ be a kernel function and let $r(s)$ be a function that is $\mathscr{O}(s^{-2})$ at infinity. Then
		\begin{equation}
			\sum_{\substack{\alpha \in O \\}} \mathrm{Res}\left(r(s)\xi(s),s=\alpha\right) + \sum_{\substack{\beta \in S \\}} \mathrm{Res}\left(r(s)\xi(s),s = \beta\right) = 0,
		\end{equation}
		where $S$ is the set of poles of $r(s)$ and $O$ is the set of poles of $\xi(s)$ that are not poles of  $r(s)$. Here $\text{Res}\left(r(s),s = \alpha\right)$ denotes the residue of $r(s)$ at $s = \alpha$. The kernel function $\xi(s)$ is meromorphic in the whole complex plane and satisfies $\xi(s) = o(s)$ over an infinite collection of circles $|s| = \rho_k$ with $\rho_k \rightarrow \infty$.
		
	\end{lem}
	
	\begin{lem}\emph{(\cite[Thm. 2.4]{XuZhao-2024})}\label{lem-2,transform}
		Let $x, y, z$ and $z'$ satisfy \eqref{notations-Ramanujan}. Given the formula $\Omega(x,e^{-y},z,z')=0$, we have the transformation formula
		\begin{align}\label{lem-for-one}
			&\Omega \left( 1-x,e^{-\pi ^2/y},yz/\pi ,\frac{1}{\pi}\left( \frac{1}{x\left( 1-x \right) z}-yz' \right) \right) =0.
		\end{align}
	\end{lem}
	\begin{lem} \emph{(\cite[Lem. 1.2]{X2018})} \label{ExpandS-Xu}
		Let $n$ be an integer, then the following Laurent expansions hold:
		\begin{align}
			&\frac{\pi}{\sin(\pi s)}=(-1)^n \left( \frac{1}{s - n} +2\sum_{k=1}^{\infty}\overline{\zeta }(2k)(s -n)^{2k-1} \right),\label{pi/sin(pis)}\\
			&\frac{\pi}{\sinh(\pi s)}=(-1)^n \left( \frac{1}{s - ni} +2\sum_{k=1}^{\infty}(-1)^k\overline{\zeta }(2k)(s -ni)^{2k-1} \right),\label{pi/sinh(pis)}\\
			&\frac{\pi}{\cos(\pi s)}=(-1)^n \left\{ \frac{1}{s - \frac{2n-1}{2}} + 2 \sum_{k=1}^{\infty} \overline{\zeta }(2k)\left( s - \frac{2n-1}{2} \right)^{2k-1} \right\},\label{pi/cos(pis)}\\
			&\frac{\pi}{\cosh(\pi s)}=(-1)^n i \left\{ \frac{1}{s - \frac{2n-1}{2}i} + 2 \sum_{k=1}^{\infty}(-1)^k
			\overline{\zeta }(2k)\left( s - \frac{2n-1}{2}i \right)^{2k-1} \right\}.\label{pi/cosh(pis)}
		\end{align}
Here $\zeta(s)$ and $\overline{\zeta }(s)$ stand for the \emph{Riemann zeta function} and \emph{alternating Riemann zeta function},
		respectively, which are defined by
		$$
		\zeta \left( s \right) :=\sum_{n=1}^{\infty}{\frac{1}{n^s}}\quad(\Re(s)>1)\quad \text{and}\quad \bar{\zeta}\left( s \right) :=\sum_{n=1}^{\infty}{\frac{\left( -1 \right) ^{n-1}}{n^s}}\quad \left( \Re \left( s \right) >0 \right).
		$$
		When $s=2m$ is  even, Euler proved the famous formula
		$$
		\zeta \left( 2m \right) =\frac{\left( -1 \right) ^{m-1}B_{2m}\left( 2\pi \right) ^{2m}}{2\left( 2m \right) !},
		$$
		where $B_{2m}\in \Q$ are Bernoulli numbers defined by the generating function$$
		\frac{x}{e^x-1}=\sum_{n=0}^{\infty}{\frac{B_n}{n!}x^n.}
		$$

All these Laurent expansions have non-empty domain (annulus) of convergence. As these functions are meromorphic (with only poles as singularities), they are analytic in the complex plane except at the poles. Therefore, there are no singularities in the annular region $r<|s-a|<R$(for any expansion center $a$, $r$ as the distance from the center $a$ to its nearest singularity and $R$ as the distance from the center $a$ to its second nearest singularity), which satisfies the convergence condition of the Laurent series. What's more, for any integer $n$, the distance $r$ from the center $a$ to its nearest singularity is positive (since the poles are discrete and have no accumulation points), and $r<R$ (the distance from the center to the second nearest singularity is greater than that to the nearest one). Thus, the annular region $r<|s-a|<R$ is non-empty.
	\end{lem}
	
	\begin{lem} \emph{(\cite[Eqs. (4.4)-(4.7)]{XZ2023})} \label{ExpandS-C}
		Let $n\in \mathbb{Z}$. We have
		\begin{align}
			\frac{\left( -1 \right) ^n}{\sin \left( \frac{1\pm i}{2}z \right)}=\frac{1\mp i}{z-\left( 1\mp i \right) n\pi}+2\sum_{k=1}^{\infty}{\frac{\overline{\zeta }\left( 2k \right)}{\pi ^{2k}}\left( \frac{1\pm i}{2} \right) ^{2k-1}\left( z-\left( 1\mp i \right) n\pi \right) ^{2k-1}},
			\\
			\frac{\left( -1 \right) ^{n-1}}{\cos \left( \frac{1\pm i}{2}z \right)}=\frac{1\mp i}{z-\left( 1\mp i \right) \tilde{n}\pi}-2\sum_{k=1}^{\infty}{\frac{\overline{\zeta }\left( 2k \right)}{\pi ^{2k}}\left( \frac{1\pm i}{2} \right) ^{2k-1}\left( z-\left( 1\mp i \right) \tilde{n}\pi \right) ^{2k-1}},
		\end{align}
		where $\tilde{n}:=n-\frac{1}{2}$.
		As in Lemma \ref{ExpandS-Xu}, since the singularities of a meromorphic function are discrete, we can always find an annular region centered at the expansion center, which lies between the nearest and second nearest singularities. This ensures that the Laurent expansion converges in this region and guarantees a non-empty convergence domain.
	\end{lem}
	\begin{lem}\emph{(\cite{DCLMRT1992})}\label{cd Maclaurin expansion}
		The Maclaurin series of ${\rm cd}(u)$ and ${\rm nd}(u)$ have the forms
		\begin{align}
			\mathrm{cd}(u)=\sum_{n=0}^{\infty}{\mathrm{S}_{2n}\left( x \right) \frac{\left( -1 \right) ^nu^{2n}}{\left(2n\right)!}}\quad\text{and}
         \quad \mathrm{nd}(u)=\sum_{n=0}^{\infty}{\mathrm{A}_{2n}\left( x \right) \frac{\left( -1 \right) ^nu^{2n}}{\left(2n\right)!}},
		\end{align}	
where $\mathrm{nd}(u):=1/\mathrm{dn}(u)$, and $\mathrm{S}_{2n}\left( x \right),\mathrm{A}_{2n}\left( x \right)\in \mathbb{Z}[x]$. Here \( x = k^2 \) is consistent with the notation used in \eqref{notations-Ramanujan}.
	\end{lem}

	\section{Berndt-Type Integrals via Hyperbolic Series}
In this section, we establish  exact correspondences between Berndt-type integrals and Ramanujan-type hyperbolic series through the contour integration techniques. Our approach reveals deep connections between these  	apparently distinct mathematical objects.

	\begin{thm}\label{thm1cosh++}
		For any positive integer $p\ge 3$, the following identity holds:
		\begin{align}
&\left(1+(-i)^{p+1}\right)\int_0^{\infty}{\frac{x^p\mathrm{d}x}{\left[ \cosh(2x)-\cos(2x)\right] \left[ \cosh x+\cos x \right]}}
			\nonumber\\
			&	=\frac{1}{2^{p+1}}ip(1-i)^{p-1}\pi^p\sum_{n=1}^{\infty}{\frac{\left( -1 \right) ^n(2n-1)^{p-1}}{\sinh\left((2n-1)\pi/2\right)\cosh ^2\left((2n-1)\pi/2\right)}}	\nonumber\\
			&\quad	+\frac{1}{8}(1-i)^{p+1}\pi^{p+1}\sum_{m=1}^{\infty}{\frac{\left( -1 \right) ^mm^p}{\sinh(m\pi)\cosh ^2(m\pi)}}-\frac{1}{2^{p+1}}i(1-i)^{p-1}\pi^{p+1}\sum_{n=1}^{\infty}{\frac{\left( -1 \right) ^n(2n-1)^p}{\cosh^3\left((2n-1)\pi/2\right)}}	\nonumber\\
			&\quad	-\frac{1}{2^{p+2}}i(1-i)^{p-1}\pi^{p+1}\sum_{n=1}^{\infty}{\frac{\left( -1 \right) ^n(2n-1)^p}{\sinh^2\left((2n-1)\pi/2\right)\cosh\left((2n-1)\pi/2\right)}}.
		\end{align}	
	\end{thm}
	\begin{proof}
		Let $z=x+iy$ for $x,y\in \R$. Consider the contour integral
		\begin{align}\label{complex equation,cosh,1/2}
			\mathscr{I} _a=\lim_{R\rightarrow \infty} \int_{\mathrm{C}_R}^{}{\frac{z^p\mathrm{d}z}{\left[ \cosh(2z)-\cos(2z)\right] \left[ \cosh z+\cos z \right]}}=\lim_{R\rightarrow \infty} \int_{\mathrm{C}_R}^{}{F\left( z \right)}\mathrm{d}z,
		\end{align}
		where $C_R$ denotes the quarter-circular contour consisting of the interval $[0,R]$, the quarter-circle $\Gamma_R$ with $|z|=R\ (-\pi/2 \le arg z\le 0)$, and the segment $[-iR,0]$, the contour is illustrated below:
		\begin{center}
			\includegraphics[height=2in]{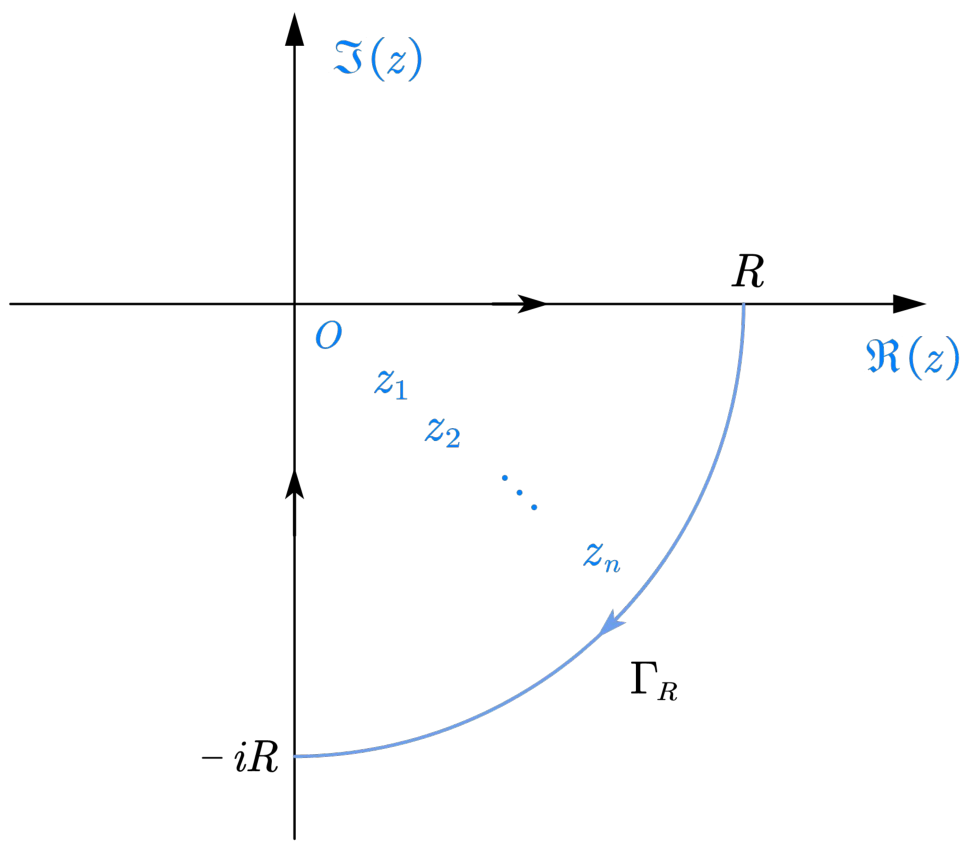}
		\end{center}
		The poles are located at
		$$
		\left[ \cosh(2z)-\cos(2z)\right] \left[ \cosh z+\cos z \right] =16\sin\left\{ \frac{1+i}{2}z \right\}\sin\left\{ \frac{1-i}{2}z \right\}\cos ^2\left\{ \frac{1+i}{2}z \right\}\cos ^2\left\{ \frac{1-i}{2}z \right\} =0.
		$$
		The poles enclosed by $C_R$ are located at $z_m=m\pi(1-i)\ (m\ge 1,\ |z_m|< R)$ and $z_n=(2n-1)\pi(1-i)/2\ (n\ge 1,\ |z_n|< R)$. Using Lemma \ref{ExpandS-C}, the residues of $\text{Res}[F(z),z]$ at these poles are given by:
		\begin{align}
			\mathrm{Res}\left[ F\left( z \right) ,z_m \right] &=\frac{\left( -1 \right) ^m(1-i)\left[ (1-i)m\pi \right] ^p}{16 \sin\left( -im\pi\right)\cos^2\left(m\pi\right)\cos^2\left(-im\pi\right)}
			=\frac{\left( -1 \right) ^{m-1}(1-i)\left[ (1-i)m\pi \right] ^p}{16 i\sinh\left( m\pi\right)\cosh^2\left(m\pi\right)},\\
			\mathrm{Res}\left[ F\left( z \right) ,z_n \right] &=(1-i)^2\left[\frac{p\left[(2n-1)\pi(1-i)/2\right]^{p-1}}{16(-1)^{n-1}\left[-i\sinh((2n-1)\pi/2)\right]\cosh^2((2n-1)\pi/2)}\right]\nonumber\\
			&\quad-(1-i)^2\left[\frac{(1-i)\left[(2n-1)\pi(1-i)/2\right]^p}{32(-1)^{n-1}\left[-i\sinh((2n-1)\pi/2)\right]^2\cosh((2n-1)\pi/2)}\right]\nonumber\\
			&\quad+(1-i)^2\left[\frac{(1-i)\left[(2n-1)\pi(1-i)/2\right]^p}{16(-1)^{n-1}\cosh^3((2n-1)\pi/2)}\right].
		\end{align}
		As $R\rightarrow \infty$, we deduce
		$$
		\int_{\Gamma _R}^{}{\frac{z^p\mathrm{d}z}{\left[ \cosh(2z)-\cos(2z)\right] \left[ \cosh z+\cos z \right]}}=o\left( 1 \right).
		$$
		Applying the residue theorem and taking the limit $R\rightarrow \infty$, we obtain
		\begin{align*}
			&-2\pi i\sum_{m=1}^{\infty}{\mathrm{Res}\left[ F\left( z \right) ,z_m \right]}-2\pi i\sum_{n=1}^{\infty}{\mathrm{Res}\left[ F\left( z \right) ,z_n \right]}=\lim_{R\rightarrow \infty} \int_{\mathrm{C}_R}^{}{\frac{z^p\mathrm{d}z}{\left[ \cosh(2z)-\cos(2z)\right] \left[ \cosh z+\cos z \right]}}
			\\
			&=\int_0^{\infty}{\frac{x^p\mathrm{d}x}{\left[ \cosh(2x)-\cos(2x)\right] \left[ \cosh x+\cos x \right]}}+i\int_0^{\infty}{\frac{\left( -ix \right) ^p\mathrm{d}x}{\left[ \cosh(-2ix)-\cos(-2ix)\right] \left[ \cosh(-ix)+\cos(-ix)\right]}}
			\\
			&=\int_0^{\infty}{\frac{x^p\mathrm{d}x}{\left[ \cosh(2x)-\cos(2x)\right] \left[ \cosh x+\cos x \right]}}+(-i)^{p+1}\int_0^{\infty}{\frac{x^p\mathrm{d}x}{\left[ \cosh(2x)-\cos(2x)\right] \left[ \cosh(x)+\cos(x)\right]}}.
		\end{align*}
		Combining these results yield the desired identity.
	\end{proof}
	
	\begin{thm}\label{Sin+,Cos-}
		For $a,b,\theta \in\mathbb{R}$ and $|\theta|<2b\pi$, $a,b \neq 0$,
		\begin{align}
			&b^2\pi\sum_{n = 1}^{\infty}\frac{(-1)^n\sinh(n \theta/a)}{\sinh(bn\pi/a)\cosh^2(bn\pi/a)}+a b\pi\sum_{n = 1}^{\infty}\frac{(-1)^n\sin(n\theta/b)}{\sinh(an\pi/b)}-a\theta\sum_{n = 1}^{\infty}\frac{(-1)^n\cos((2n-1)\theta/(2b))}{\sinh((2n-1)a\pi/(2b))}\nonumber\\
			&+a^2\pi\sum_{n = 1}^{\infty}\frac{(-1)^n\sin((2n-1)\theta/(2b))\cosh((2n-1)a\pi/(2b))}{\sinh^2((2n-1)a\pi/(2b))}+\frac{\theta b}{2}=0,\label{Z1th}\\
			&b^2\pi\sum_{n = 1}^{\infty}\frac{(-1)^n\cosh((2n-1)\theta/(2a))}{\sinh((2n-1)b\pi/(2a))\cosh^2((2n-1)b\pi/(2a))}+a\theta\sum_{n = 1}^{\infty}\frac{(-1)^n\sin((2n-1)\theta/(2b))}{\cosh((2n-1)a\pi/(2b))}\nonumber\\
			&+a^2\pi\sum_{n = 1}^{\infty}\frac{(-1)^n\cos((2n-1)\theta/(2b))\sinh((2n-1)a\pi/(2b))}{\cosh^2((2n-1)a\pi/(2b))}+a b\pi\sum_{n = 1}^{\infty}\frac{(-1)^n\cos(n\theta/b)}{\cosh(an\pi/b)}\nonumber\\&\quad+\frac{ab\pi}{2}=0,\label{Z2th}\\
			&b^2\pi\sum_{n = 1}^{\infty}\frac{(-1)^n\sinh((2n-1)\theta/(2a))}{\sinh^2((2n-1)b\pi/(2a))\cosh((2n-1)b\pi/(2a))}+a\theta\sum_{n = 1}^{\infty}\frac{(-1)^n\cos(n\theta/b)}{\cosh(an\pi/b)}\nonumber\\
			&-a^2\pi\sum_{n = 1}^{\infty}\frac{(-1)^n\sin(n\theta/b)\sinh(an\pi/b)}{\cosh^2(an\pi/b)}+a b\pi\sum_{n = 1}^{\infty}\frac{(-1)^n\sin((2n-1)\theta/(2b))}{\cosh((2n-1)a\pi/(2b))}+\frac{a\theta}{2}=0.\label{Z3th}
		\end{align}	
	\end{thm}
	\begin{proof}
		This proof relies on analyzing some functions and the usual residue computation. These functions are meromorphic in the entire complex plane with some simple poles at specific points.

		First, we note that function \(Z_1:=Z(s,\theta;a,b|1,0,1,2)\) has poles at
$s=\pm\frac{n}{a}$ (simple poles),\ $s=\pm\frac{ni}{b}$ (simple poles),
$s=\pm\frac{(2n-1)i}{2b}$ (double poles) and $s=0$ (simple poles),
where \(n\in\mathbb{N}\).\\
		By using Lemma \ref{ExpandS-Xu}, for all \(n\in\mathbb{Z}\) we may deduce the asymptotic expansions of a few reciprocal quadratic trigonometric and hyperbolic functions as follows
		\begin{align}\label{(pi/cosh^2(bpis))^2}
			\left(\frac{\pi}{\cosh(b\pi s)}\right)^2&=-\frac{1}{\left(bs-\frac{2n - 1}{2}i\right)^2}+2\zeta(2)-6\zeta(4)\left(bs-\frac{2n - 1}{2}i\right)^2\nonumber\\
			&\quad+o\left(\left(bs-\frac{2n - 1}{2}i\right)^2\right).
		\end{align}
After direct computation, the residues of \( Z_1 \) at the simple poles \( \pm\frac{n}{a} \) and \( s = \pm\frac{ni}{b} \), as well as the residues at the second-order poles \( \pm\frac{(2n-1)i}{2b} \) derived by applying \eqref{(pi/cosh^2(bpis))^2}, are respectively as follows:
\begin{align}
			&\mathrm{Res}\left[ Z_1,s=\pm\frac{n}{a}\right] =\frac{\pi^3}{a}\frac{(-1)^n\sinh(n\theta/a)}{\sinh(b n\pi/a)\cosh^2(b n\pi/a)},\label{z1s=a}\\
			&\mathrm{Res}\left[ Z_1,s=\pm\frac{ni}{b}\right]=\frac{\pi^3}{b}\frac{(-1)^n\sin(n\theta/a)}{\sinh(a n\pi/b)},\label{z1s=ni/b}\\
			&\mathrm{Res}\left[ Z_1,s=\pm\frac{(2n-1)i}{2b}\right]=-\frac{\pi^2\theta}{b^2}\frac{(-1)^n\cos((2n-1)\theta/(2b))}{\sinh((2n-1)a \pi/(2b))}\nonumber\\
			&\qquad\qquad\qquad\qquad\qquad\qquad+\frac{a\pi^3}{b^2}\frac{(-1)^n \sin((2n-1)\theta/(2b))\cosh((2n-1)a\pi/(2b))}{\sinh^2((2n-1)a\pi/(2b))}.\label{z1s=(2n-1)i/2b}
		\end{align}
		Furthermore, by applying equations (\ref{pi/sin(pis)}) and (\ref{pi/sinh(pis)}), along with the power series expansion of the hyperbolic sine function, we deduce the following expansion:
		\begin{align}
			Z_1=\frac{\pi^2}{\cosh^2(b\pi s)}\left(\frac{\theta}{ab}\frac{1}{s}+\frac{\theta^3}{6ab}s-\left(\frac{b}{a}-\frac{a}{b}\right) \zeta(2)\theta s+o(1)\right),\quad s\to 0.\nonumber
		\end{align}
Then, using the above expansion, we can compute the residue of $Z_1$ at 0 as
		\begin{align}\label{z1s=0}
			&\mathrm{Res}\left[ Z_1,s=0\right] =\frac{\theta\pi^2}{ab}.
		\end{align}
		By Lemma \ref{Sum Res=0}, summing the four contributions (\ref{z1s=a})-(\ref{z1s=0}) yields the desired result (\ref{Z1th}).
		
		Secondly, the function \(Z_2:=\bar{Z}(s,\theta;a,b|0,1,1,2)\) possesses poles at
$s=\pm\frac{2n-1}{2a}$ (simple poles),\ $s=\pm\frac{ni}{b}$ (simple poles),\ $s=\pm\frac{(2n-1)i}{2b}$ (double poles) and $s=0$ (simple poles),
where \(n\in\mathbb{N}\).
Through direct computation, we obtain the following residue values of \( Z_2 \) at the simple poles \( \pm\frac{2n-1}{2a} \) and \( \pm\frac{ni}{b} \), as well as at the second-order poles \( \pm\frac{(2n-1)i}{2b} \):
		\begin{align}
			&\mathrm{Res}\left[ Z_2,s=\pm\frac{2n-1}{2a}\right] =\frac{\pi^3}{a}\frac{(-1)^n\cosh((2n-1)\theta/(2a))}{\sinh((2n-1)b\pi/(2a))\cosh^2((2n-1)b\pi/(2a))},\label{z2s=(2n-1)/2a}\\
&\mathrm{Res}\left[ Z_2,s=\pm\frac{ni}{b}\right]=\frac{\pi^3}{b}\frac{(-1)^n\cos(n\theta/b)}{\cosh(a n\pi/b)},\label{z2s=ni/b}\\
&\mathrm{Res}\left[ Z_2,s=\pm\frac{(2n-1)i}{2b}\right]=\frac{\pi^2\theta}{b^2}\frac{(-1)^n\sin((2n-1)\theta/(2b))}{\cosh((2n-1)a \pi/(2b))}\nonumber\\
			&\qquad\qquad\qquad\qquad\qquad+\frac{a\pi^3}{b^2}\frac{(-1)^n \cos((2n-1)\theta/(2b))\sinh((2n-1)a\pi/(2b))}{\cosh^2((2n-1)a\pi/(2b))}.\label{z2s=(2n-1)i/2b}
		\end{align}
		When $s=0$, using the expansions from (\ref{pi/sinh(pis)}) and (\ref{pi/cos(pis)}), we derive
		\begin{align}
			Z_2(s,\theta;a,b)=\frac{\pi^3\cosh(\theta s)}{\cos(a\pi s)\cosh^2(b\pi s)}\left(\frac{1}{bs}-\zeta(2)bs+o(1)\right),\quad s\to 0,\nonumber
		\end{align}
		which gives the residue
		\begin{align}\label{z2s=0}
			&\mathrm{Res}\left[ Z_2,s=0\right] =\frac{\pi^3}{b}.
		\end{align}
		By Lemma \ref{Sum Res=0}, the sum of these four residue contributions (\ref{z2s=(2n-1)/2a})-(\ref{z2s=0}) establishes the desired identity (\ref{Z2th}).

		At last, the function \(Z_3:=Z(s,\theta;a,b|0,1,2,1)\) admits poles at $s=\pm\frac{2n-1}{2a}$ (simple poles),\ $s=\pm\frac{(2n-1)i}{2b}$ (simple poles),\ $s=\pm\frac{ni}{b}$ (double poles),\ $s=0$ (simple poles),\ where \(n\in\mathbb{N}\). We apply the expansion given in Lemma \ref{ExpandS-Xu} to obtain
		\begin{align}\label{revisond-addedex-bihao}
			\left(\frac{\pi}{\sinh(b\pi s)}\right)^2&=\frac{1}{\left(bs-ni\right)^2}-2\zeta(2)+6\zeta(4)\left(bs-ni\right)^2+o\left(\left(bs-ni\right)^2\right)\quad\text{as }s\rightarrow ni.
		\end{align}
Therefore, by performing direct calculations and applying \eqref{revisond-addedex-bihao}, the residues of \( Z_3 \) at the simple poles \( \pm\frac{2n-1}{2a} \) and \( \pm\frac{(2n-1)i}{2b} \), as well as at the second-order poles \( \pm\frac{ni}{b} \), can be readily obtained as follows:
\begin{align}
			&\mathrm{Res}\left[ Z_3,s=\pm\frac{2n-1}{2a}\right] =\frac{\pi^3}{a}\frac{(-1)^n\sinh((2n-1)\theta/(2a))}{\sinh^2((2n-1)b\pi/(2a))\cosh((2n-1)b\pi/(2a))},\label{z3s=(2n-1)/2a}\\
			&\mathrm{Res}\left[ Z_3,s=\pm\frac{(2n-1)i}{2b}\right]=\frac{\pi^3}{b}\frac{(-1)^n\sin((2n-1)\theta/(2b))}{\cosh((2n-1)a\pi/(2b))},\label{z3s=(2n-1)i/2b}\\
			&\mathrm{Res}\left[ Z_3,s=\pm\frac{ni}{b}\right]=\frac{\pi^2\theta}{b^2}\frac{(-1)^n\cos(n\theta/b)}{\cosh(an\pi/b)}-\frac{a\pi^3}{b^2}\frac{(-1)^n \sin(n\theta/b)\sinh(an\pi/b)}{\cosh^2(an\pi/b)}.\label{z3s=ni/b}
		\end{align}
Applying \eqref{revisond-addedex-bihao} (with \( n = 0 \)) and the power series expansion of the hyperbolic sine function, the following expansion of \( Z_3 \) at zero can be computed:
		\begin{align}
			Z_3=\frac{\pi^2}{\cos(a\pi s)\cosh(b\pi s)}\left(\frac{\theta}{b^2}\frac{1}{s}+\frac{\theta^3}{6b^2}s-2\zeta(2)\theta s+o(1)\right),\quad s\to 0,\nonumber
		\end{align}
		from which we extract the residue
		\begin{align}\label{z3s=0}
			&\mathrm{Res}\left[ Z_3,s=0\right] =\frac{\theta\pi^2}{b^2}.
		\end{align}
		Applying Lemma \ref{Sum Res=0} to the total of four residue terms (\ref{z3s=(2n-1)/2a})-(\ref{z3s=0}) establishes identity (\ref{Z3th}).

		The three identities are derived by systematically computing residues of $Z_1,Z_2,Z_3$ and applying the residue theorem. This completes the proof of Theorem \ref{Sin+,Cos-}.
	\end{proof}

	\begin{thm}
		Let $p>1$ be an odd integer, we have the following equation established
		\begin{align}
			&\sum_{n = 1}^{\infty}\frac{(-1)^n n^p}{\sinh(ny)\cosh^2(ny)}\nonumber\\
&\quad\quad\quad\quad\quad\quad\quad=-\frac{\pi^{p+1}}{y^{p+1}}(-1)^{(p-1)/2}\sum_{n = 1}^{\infty}\frac{(-1)^n n^p}{\sinh(\pi^2n/y)}\nonumber\\
			&\quad\quad\quad\quad\quad\quad\quad\quad+\frac{p\pi^p}{2^{p-1}y^{p+1}}(-1)^{(p-1)/2}\sum_{n = 1}^{\infty}\frac{(-1)^n(2n-1)^{p-1}}{\sinh((2n-1)\pi^2/(2y))}\nonumber\\
			&\quad\quad\quad\quad\quad\quad\quad\quad-\frac{\pi^{p+2}}{2^py^{p+2}}(-1)^{(p-1)/2}\sum_{n = 1}^{\infty}\frac{(-1)^n(2n-1)^p\cosh((2n-1)\pi^2/(2y))}{\sinh^2((2n-1)\pi^2/(2y))},\label{js1}\\
			&\sum_{n = 1}^{\infty}\frac{(-1)^n(2n-1)^{p-1}}{\sinh((2n-1)y/2)\cosh^2((2n-1)y/2)}\nonumber\\
&\quad\quad\quad\quad\quad\quad\quad=-\frac{2^{p-1}\pi^p}{y^p}(-1)^{(p-1)/2}\sum_{n = 1}^{\infty}\frac{(-1)^n n^{p-1}}{\cosh(\pi^2n/y)}\nonumber\\
			&\quad\quad\quad\quad\quad\quad\quad\quad+\frac{2(p-1)\pi^{p-1}}{y^p}(-1)^{(p-1)/2}\sum_{n = 1}^{\infty}\frac{(-1)^n(2n-1)^{p-2}}{\cosh((2n-1)\pi^2/(2y))}\nonumber\\
			&\quad\quad\quad\quad\quad\quad\quad\quad-\frac{\pi^{p+1}}{y^{p+1}}(-1)^{(p-1)/2}\sum_{n = 1}^{\infty}\frac{(-1)^n(2n-1)^{p-1}\sinh((2n-1)\pi^2/(2y))}{\cosh^2((2n-1)\pi^2/(2y))},\label{js2}\\
			&\sum_{n = 1}^{\infty}\frac{(-1)^n(2n-1)^p}{\sinh^2((2n-1)y/2)\cosh((2n-1)y/2)}\nonumber\\
&\quad\quad\quad\quad\quad\quad\quad\quad\quad\quad\quad\quad\quad\quad=-\frac{\pi^{p+1}}{y^{p+1}}(-1)^{(p-1)/2}\sum_{n = 1}^{\infty}\frac{(-1)^n(2n-1)^p}{\cosh((2n-1)\pi^2/(2y))}\nonumber\\
			&\quad\quad\quad\quad\quad\quad\quad\quad\quad\quad\quad\quad\quad\quad\quad-\frac{p2^p\pi^p}{y^{p+1}}(-1)^{(p-1)/2}\sum_{n = 1}^{\infty}\frac{(-1)^n n^{p-1}}{\cosh(\pi^2n/y)}\nonumber\\
			&\quad\quad\quad\quad\quad\quad\quad\quad\quad\quad\quad\quad\quad\quad\quad+\frac{2^p\pi^{p+2}}{y^{p+2}}(-1)^{(p-1)/2}\sum_{n = 1}^{\infty}\frac{(-1)^n n^p\sinh(\pi^2n/y)}{\cosh^2(\pi^2n/y)}.\label{js3}
		\end{align}
	\end{thm}
	\begin{proof}
		We prove the first identity in detail, as the proofs for the other two follow similarly. By taking the $p$-th odd derivative of (\ref{Z1th}) in Theorem \ref{Sin+,Cos-} with respect to $\theta$, then $\theta\rightarrow 0$, we acquire the result as follows
		\begin{align*}
			&
			\frac{\mathrm{d}^p}{\mathrm{d}\theta ^p}\left[ \sinh \left( \frac{n\theta}{a} \right) \right] _{|\theta =0}=\left( \frac{n}{a} \right) ^p
		,\quad
			\frac{\mathrm{d}^p}{\mathrm{d}\theta^p}\left[\theta\cos\left(\frac{2n-1}{2b}\theta\right)\right]_{|\theta =0}=p(-1)^{(p-1)/2}\left(\frac{2n-1}{2b}\right)^{p-1},\\
			&\frac{\mathrm{d}^p}{\mathrm{d}\theta^p}\left[\sin\left(\frac{n\theta}{b}\right)\right]_{|\theta =0}=(-1)^{(p-1)/2}\left(\frac{n}{b}\right)^p,\quad
			\frac{\mathrm{d}^p}{\mathrm{d}\theta^p}\left[\sin\left(\frac{2n-1}{2b}\theta\right)\right]_{|\theta =0}=(-1)^{(p-1)/2}\left(\frac{2n-1}{2b}\right)^p.
		\end{align*}
		Using the above derivatives and letting $y= \frac{b\pi}{a} $ in (\ref{Z1th}), and we can easily obtain that
		\begin{align}
			&\frac{b}{2}\delta_p-(-1)^{(p-1)/2}\frac{ap}{2^{p-1}b^{p-1}}\sum_{n=1}^{\infty}\frac{(-1)^n (2n-1)^{p-1}}{\sinh((2n-1)\pi^2/(2y))}\nonumber\\
			&\quad\quad\quad+\frac{b^2\pi}{a^{p}}\sum_{n=1}^{\infty}\frac{(-1)^n n^p}{\sinh(ny)\cosh^2(ny)}+(-1)^{(p-1)/2}\frac{a\pi}{b^{p-1}}\sum_{n=1}^{\infty}\frac{(-1)^n n^p}{\sinh(n\pi^2/y)}\nonumber\\
&\quad\quad\quad+(-1)^{(p-1)/2}\frac{a^2\pi}{2^p b^p}\sum_{n=1}^{\infty}\frac{(-1)^n(2n-1)^p\cosh((2n-1)\pi^2/(2y))}{\sinh^2((2n-1)\pi^2/(2y))}=0,
		\end{align}
		where $\delta_1=1$ and $\delta_p = 0$ if $p\geq3 $.
		
For $ p=1 $, we derive the special case
		\begin{align}
			\sum_{n = 1}^{\infty}\frac{(-1)^n n}{\sinh(ny)\cosh^2(ny)}=&-\frac{1}{2y}-\frac{\pi^2}{y^2}\sum_{n = 1}^{\infty}\frac{(-1)^n n}{\sinh(\pi^2n/y)}\nonumber\\
			&+\frac{\pi}{y^2}\sum_{n = 1}^{\infty}\frac{(-1)^n}{\sinh((2n-1)\pi^2/(2y))}\nonumber\\
			&-\frac{\pi^3}{2y^3}\sum_{n = 1}^{\infty}\frac{(-1)^n(2n-1)\cosh((2n-1)\pi^2/(2y))}{\sinh^2((2n-1)\pi^2/(2y))},
		\end{align}
and for $ p\geq 3 $, we deduce (\ref{js1}) .

		The proofs of (\ref{js2}) and (\ref{js3}) follow similarly, with the following key differences: For (\ref{js2}), we take the $(p-1)$-th order derivative (i.e., even-order derivative) of equation (\ref{Z2th}) with respect to $\theta$. For (\ref{js3}), we employ the $p$-th order odd derivative of equation (\ref{Z3th}) with respect to $\theta$.
	\end{proof}

	\section{Evaluations of Hyperbolic Summations via Jacobi Functions}
	
	We now evaluate the following reciprocal hyperbolic series in the Ramanujan type
	\begin{align}
		G_{p,1,2}(y)&:=\sum_{n = 1}^{\infty}\frac{(-1)^n n^p}{\sinh(ny)\cosh^2(ny)},\\
		\tilde{G}_{p,1,2}(y)&:=\sum_{n = 1}^{\infty}\frac{(-1)^n (2n-1)^{p}}{\sinh((2n-1)y/2)\cosh^2((2n-1)y/2)},\\
		\tilde{G}_{p,2,1}(y)&:=\sum_{n = 1}^{\infty}\frac{(-1)^n(2n-1)^p}{\sinh^2((2n-1)y/2)\cosh((2n-1)y/2)}
	\end{align}
	by applying the Fourier series expansions and the Maclaurin series expansions of relevant Jacobi elliptic functions. Let $x,y$ and $z$ satisfy the relations in \eqref{notations-Ramanujan}.
	
	\begin{thm}\label{barC2p,2}
		Let $p\ge 3$ be an odd integer. We have
		\begin{align}\label{hyperseries-one}
			G_{p,1,2}(y)=&(-1)^{(p-1)/2}\frac{(p-1)!z^{p+1}z'x(1-x)}{2^{p+1}}\mathrm{R}_{p-1}(1-x)-\frac{pz^{p+1}z'x(1-x)}{2^p}\mathrm{S}_{p-1}(1-x)
			\sqrt{1-x}\nonumber\\
			&-\frac{z^{p+2}x(1-x)}{2^{2p}}\frac{\mathrm{d}}{\mathrm{d}x}\left[ \mathrm{S}_{p-1}(1-x)\sqrt{1-x} \right],
		\end{align}
		where $\mathrm{R}_{p-1}(1-x)=\frac{(-x)^{(p-3)/2}}{(p-1)!}Q_{p-1}\left(\frac{1-x}{-x}\right)\in\mathbb{Q}[x]$ and ${Q}_{p-1}\left( x \right)$ represents coefficients in the Maclaurin expansion of $\mathrm{sn}^2(u)$;
		${\rm S}_{p-1}\left( x \right)$ represents coefficients in the Maclaurin expansion of $\mathrm{cd}\left( u \right)$.
	\end{thm}
	\begin{proof}
		Beginning with (\ref{js1}), we derive
		\begin{align}\label{equ-lemma3-one}
			G_{p,1,2}(y) &=-(-1)^{(p-1)/2}\frac{\pi^{p+1}}{y^{p+1}}G_{p,1,0}\left(\frac{\pi^2}{y}\right)+(-1)^{(p-1)/2}\frac{p\pi ^p}{2^{p-1}y^{p+1}}\tilde G_{p-1,1,0}\left(\frac{\pi^2}{y}\right)\nonumber\\
			&\quad-(-1)^{(p-1)/2}\frac{\pi^{p+2}}{2^p y^{p+2}}\tilde G_{p,2,-1}\left(\frac{\pi^2}{y}\right).
		\end{align}
		Now, we need to compute the series $\tilde G_{p-1,1,0}\left( y \right)$. According to Lemma \ref{cd Maclaurin expansion}, we obtain the Maclaurin expansion for ${\rm cd}(u)$
		\begin{align}\label{cd-exped-one}
			\mathrm{cd}\left( u \right) =\sum_{n=0}^{\infty}{\mathrm{S}_{2n}\left( x \right) \frac{\left( -1 \right) ^nu^{2n}}{\left( 2n \right) !}}.
		\end{align}	
		Using known results from \cite{DCLMRT1992}, we record the Fourier series expansions of Jacobi elliptic functions ${\rm cd}(u)$
		\begin{align}\label{cd}
			&{\rm cd}(u)=\frac{2\pi}{Kk}\sum_{n=0}^{\infty}\frac{(-1)^n q^{n+1/2}}{1-q^{2n+1}}\cos\left[(2n+1)\frac{\pi u}{2K}\right]\quad (|q|<1).
		\end{align}
		Applying $q\equiv q\left( x \right):=e^{-y}$ to \eqref{cd}, we obtain
		\begin{align}\label{cd-one}
			{\rm cd}(u)&=\frac{2\pi}{{K}k}\sum_{n=0}^{\infty}{\frac{(-1)^n q^{n+\tfrac{1}{2}}\left( 2n+1 \right) ^{2j}}{1-q^{2n+1}}}\sum_{j=0}^{\infty}{\frac{\left( -1 \right) ^j\left( \frac{\pi u}{2{K}} \right) ^{2j}}{\left( 2j \right) !}}
			\nonumber\\
			&=\frac{\pi}{{K}k}\sum_{j=0}^{\infty}{\frac{\left( -1 \right) ^j\left( \frac{\pi u}{2{K}} \right) ^{2j}}{\left( 2j \right) !}}\tilde G_{2j,1,0}\left( y \right) .
		\end{align}
And in the p. $165$ of \emph{Ramanujan's Notebooks (III)}\cite{B1991}, we obtain
		\begin{align}
			\frac{1}{2}z\sqrt{x}{\rm cd}(u)=\sum_{n=0}^{\infty}\frac{(-1)^n\cos\left[(2n+1)u\right]}{\sinh((2n+1)y/2)}.
		\end{align}	
		Comparing the coefficients of $u^{2n}$ in \eqref{cd-exped-one} and \eqref{cd-one}, we deduce
		\begin{align}
			\tilde G_{2n,1,0}\left( y \right)=\sum_{j=1}^{\infty}{\frac{(-1)^j\left( 2j-1 \right) ^{2n}}{\sinh \left( \frac{2j-1}{2}y \right)}}
			=(-1)^n z^{2n+1}\sqrt{x}\frac{\mathrm{S}_{2n+1}\left( x \right)}{2}. \label{Y2n+1,1}
		\end{align}
Fix $\delta>0$. For $y\ge\delta$ we have
\[
\left|\frac{(-1)^j(2j-1)^{2n}}{\sinh\!\left(\frac{2j-1}{2}y\right)}\right|
\le
\frac{(2j-1)^{2n}}{\sinh\!\left(\frac{2j-1}{2}\delta\right)}
\le
2(2j-1)^{2n}e^{-\tfrac{2j-1}{2}\delta},
\]
hence $\widetilde G_{2n,1,0}(y)$ converges uniformly on $[\delta,\infty)$ by the Weierstrass M-test.
Moreover,
\[
f_j(y):=\frac{(-1)^j(2j-1)^{2n}}{\sinh\!\left(\frac{2j-1}{2}y\right)}
\quad\Rightarrow\quad
f_j'(y)= -\frac{(-1)^j(2j-1)^{2n+1}}{2}\,
\frac{\cosh\!\left(\frac{2j-1}{2}y\right)}{\sinh^{2}\!\left(\frac{2j-1}{2}y\right)}.
\]
Using $\frac{\cosh t}{\sinh^2 t}\le C e^{-t}$ for  $t=\frac{2j-1}{2}y\ge \frac{2j-1}{2}\delta$, we obtain
\[
|f_j'(y)|\le C(2j-1)^{2n+1}e^{-\tfrac{2j-1}{2}\delta},
\]
so the derivative series $\sum_{j\ge1} f_j'(y)$ converges uniformly on $[\delta,\infty)$.
		Hence, differentiating \eqref{Y2n+1,1} gives the relation
		\begin{align*}
			\frac{\mathrm{d}}{\mathrm{d}y}\tilde G_{p-1,1,0}\left( y \right) =\sum_{n=1}^{\infty}{\frac{\mathrm{d}}{\mathrm{d}y}\frac{(-1)^n\left( 2n-1 \right) ^{p-1}}{\sinh \left( \frac{2n-1}{2}y \right)}=-\frac{1}{2}}\tilde G_{p,2,-1}\left( y \right).
		\end{align*}
	By applying the transformation $\frac{\mathrm{d}x}{\mathrm{d}y}=-x(1-x)z^2$ (\cite[p. 120, Entry. 9(i)]{B1991} or \eqref{dx/dy}), we obtain
		\begin{align}\label{equ-relation-differ-one}
			\tilde G_{p,2,-1}\left( y \right)=2x\left( 1-x \right) z^2\frac{\mathrm{d}}{\mathrm{d}x}\tilde G_{p-1,1,0}\left( y \right).
		\end{align}
		Resorting to Lemma \ref{lem-2,transform}, for $\Omega \left( x,e^{-y},z,z' \right) =0$, we have
		$$\Omega \left( 1-x,e^{-\pi ^2/y},yz/\pi ,\frac{1}{\pi}\left( \frac{1}{x\left( 1-x \right) z}-yz' \right) \right) =0.$$
		Hence, the \eqref{Y2n+1,1} and \eqref{equ-relation-differ-one} can be rewritten as
		\begin{align}\label{equ-lemma3-two}
			\tilde G_{p-1,1,0}\left( \frac{\pi ^2}{y} \right) =-(-1)^{(p-1)/2}\left( \frac{yz}{\pi} \right) ^{p}\frac{\mathrm{S}_{p-1}\left( 1-x \right)}{2}\sqrt{1-x},
		\end{align}
		and
		\begin{align}\label{equ-lemma3-three}
			\tilde G_{p,2,-1}\left( \frac{\pi ^2}{y} \right)&=2x\left( 1-x \right) \left( \frac{yz}{\pi} \right) ^2\frac{\mathrm{d}}{\mathrm{d}\left( 1-x \right)}\tilde G_{p-1,1,0}\left( \frac{\pi ^2}{y} \right)\nonumber\\
			&=-2x\left( 1-x \right) \left( \frac{yz}{\pi} \right) ^2\frac{\mathrm{d}}{\mathrm{d}x}\tilde G_{p-1,1,0}\left( \frac{\pi ^2}{y} \right).
		\end{align}
		What's more, as known from Rui-Xu-Zhao's paper\cite[Thm. 3.14]{RXZ2023} and apply Lemma \ref{lem-2,transform}, for $\Omega \left( x,e^{-y},z,z' \right) =0$, we deduce\\
		\begin{align}\label{Rui-one}
			\sum_{n=1}^{\infty}\frac{(-1)^n n^p}{\sinh\left(\frac{\pi^2}{y}\right)}=-\frac{(p-1)!}{2^{p+1}}\left(\frac{yz}{\pi}\right)^2 x(1-x)\mathrm{R}_{p-1}(1-x),
		\end{align}
		where $\mathrm{R}_{p-1}(1-x)=\frac{(-x)^{(p-3)/2}}{(p-1)!}Q_{p-1}\left(\frac{1-x}{-x}\right)\in\mathbb{Q}[x]$.

		Finally, combining \eqref{equ-lemma3-one}, \eqref{Y2n+1,1}, \eqref{equ-lemma3-two}, \eqref{equ-lemma3-three} and \eqref{Rui-one}, we thus complete the proof of Theorem \ref{barC2p,2}.
	\end{proof}
	\begin{exa}\label{exa-added-one}
		Using \emph{Mathematica} to evaluate specific instances of the right-hand side of \eqref{hyperseries-one}, we derive the following examples:
		\begin{align*}
			&\sum_{n = 1}^{\infty}\frac{(-1)^n n^{3}}{\sinh(ny)\cosh^2(ny)}=\frac{1}{16}xz^4\left\{6z'x (1-x)^{3/2}-\sqrt{1-x}(3 x-2)z+2 (x-1)\right\},\\
			&\sum_{n = 1}^{\infty}\frac{(-1)^n n^{7}}{\sinh(ny)\cosh^2(ny)}\\
&=\frac{1}{256}xz^8\Big\{14z'x\left(61 x^2-76 x+16\right)(1-x)^{3/2}\Big.+16(x-1)\left(17 x^2-17 x+2\right)\atop-\sqrt{1-x}(427 x^3-746 x^2+352x-32)z\Big\},\\
			&\sum_{n = 1}^{\infty}\frac{(-1)^n n^{11}}{\sinh(ny)\cosh^2(ny)}\\
&=\frac{1}{4096}xz^{12}\Big\{22z'x\left(50521x^4-113672x^3+79728x^2-16832 x+256\right)(1-x)^{3/2}\Big.\\&\Big.\qquad\qquad\qquad+256(x-1)\left(1382 x^4-2764 x^3+1641 x^2-259 x+2\right)
			\Big.\\&\Big.\quad-\sqrt{1-x}\left(555731 x^5-1528258 x^4+1467472 x^3-562528 x^2+68096 x-512\right)z\Big\},\\
            &\sum_{n = 1}^{\infty}\frac{(-1)^n n^{15}}{\sinh(ny)\cosh^2(ny)}\\
			&=\frac{1}{65536}xz^{16}\Big\{30z'x\left(199360981x^6-647923188 x^5+775638816 x^4\atop-408850432 x^3+85975296 x^2-4205568 x+4096\right)(1-x)^{3/2}\Big.\\&\Big.\quad\qquad\qquad\qquad+2048 (x-1)\left(929569 x^6-2788707 x^5+3021099 x^4\atop-1394353 x^3+240594 x^2-8202 x+4\right)\Big.\\&\Big.\quad-\sqrt{1-x}\left(2990414715 x^7-11214055178 x^6+16307105232 x^5-11436042048 x^4\atop+3872630528 x^3-536879616 x^2+16834560 x-8192\right)z\Big\}.
		\end{align*}
	\end{exa}
	
Setting $ p=4m-1$ and $x=1/2$ in Theorem \ref{barC2p,2} yields the following corollary.	
	\begin{cor}
		Set $\Gamma=\Gamma(1/4)$ and $\mathrm{S}'_p(x)=\mathrm{d}\mathrm{S}_p(x)/\mathrm{d}x$. For any integer $m>0$, we obtain
		\begin{align}\label{js1-pi}
			\sum_{n = 1}^{\infty}\frac{(-1)^n n^{4m-1}}{\sinh(n\pi)\cosh^2(n\pi)}&=-\frac{1}{2^{8m+2}}(4m-2)!\mathrm{R}_{4m-2}\left(\frac{1}{2}\right)\frac{\Gamma ^{8m}}{\pi^{6m}}\nonumber\\
			&\quad-\frac{1}{2^{8m-1}}(4m-1)\mathrm{S}_{4m-2}\left(\frac{1}{2}\right)\frac{1}{\sqrt{2}}\frac{\Gamma ^{8m-2}}{\pi^{(12m-1)/2}}\nonumber\\
			&\quad-\frac{1}{2^{8m+2}}\left(\mathrm{S}_{4m-2}'\left(\frac{1}{2}\right)-\mathrm{S}_{4m-2}\left(\frac{1}{2}\right)\right)\frac{1}{\sqrt{2}}\frac{\Gamma ^{8m+2}}{\pi^{(12m+3)/2}}.
		\end{align}
	\end{cor}
	
	\begin{exa}
		Set $\Gamma=\Gamma(1/4)$. By substituting \( y = \pi \) and \( x = 1/2 \) into Example \ref{exa-added-one} and evaluating the result with \emph{Mathematica}, we obtain the following examples:
		\begin{align*}
			\sum_{n = 1}^{\infty}\frac{(-1)^n n^{3}}{\sinh(n\pi)\cosh^2(n\pi)}&=\frac{\Gamma^{10}}{2048 \sqrt{2} \pi ^{15/2}}-\frac{\Gamma^8}{512 \pi ^6}+\frac{3 \Gamma^6}{256 \sqrt{2} \pi ^{11/2}},\\
			\sum_{n = 1}^{\infty}\frac{(-1)^n n^{7}}{\sinh(n\pi)\cosh^2(n\pi)}&=-\frac{87 \Gamma^{18}}{2097152 \sqrt{2} \pi ^{27/2}}+\frac{9 \Gamma ^{16}}{65536 \pi ^{12}}-\frac{189 \Gamma^{14}}{262144 \sqrt{2} \pi ^{23/2}},\\
			\sum_{n = 1}^{\infty}\frac{(-1)^n n^{11}}{\sinh(n\pi)\cosh^2(n\pi)}&=\frac{57969 \Gamma^{26}}{2147483648 \sqrt{2} \pi ^{39/2}}-\frac{189 \Gamma^{24}}{2097152 \pi ^{18}}+\frac{126819 \Gamma^{22}}{268435456 \sqrt{2} \pi ^{35/2}},\\
			\sum_{n = 1}^{\infty}\frac{(-1)^n n^{15}}{\sinh(n\pi)\cosh^2(n\pi)}&=-\frac{160692903 \Gamma^{34}}{2199023255552 \sqrt{2} \pi ^{51/2}}+\frac{130977 \Gamma^{32}}{536870912 \pi ^{24}}-\frac{351673245 \Gamma^{30}}{274877906944 \sqrt{2} \pi ^{47/2}}.
		\end{align*}
	\end{exa}
	
	\begin{thm}{\label{barS2p,2}}
		Let $p\ge 3$ be an odd integer. We have
		\begin{align}\label{barS2p3}
			\tilde G_{p-1,1,2}\left( y \right)& =-\frac{1}{2}z^p\mathrm{A}_{p-1}(1-x)\sqrt{x}+(p-1)z^pz'x(1-x)\mathrm{P}_{p-2}(1-x)\sqrt{x(1-x)}\nonumber\\
			&\quad+x(1-x)z^{p+1}\frac{\mathrm{d}}{\mathrm{d}x}\left[\mathrm{P}_{p-2}(1-x)\sqrt{x(1-x)}\right],
		\end{align}
		where $\mathrm{A}_{p-1}\left(x\right)$ are coefficients in the Maclaurin expansion of $\mathrm{nd}\left( u \right)$ and $\mathrm{P}_{p}\left(x\right)$ are coefficients in the Maclaurin expansion of $\mathrm{sd}\left( u \right)$ (see \cite[Lem. 3.1]{RXZ2023}).
	\end{thm}
	\begin{proof}
		To establish the result, we start with (\ref{js2})
		\begin{align}\label{thm5.3-proof-begin-1}
			\tilde G_{p-1,1,2}\left( y \right) &=-(-1)^{(p-1)/2}\frac{2^{p-1}\pi ^p}{y^p}G_{p-1,0,1}\left( \frac{\pi ^2}{y} \right) +(-1)^{(p-1)/2}\frac{2(p-1)\pi ^{p-1}}{y^p}\tilde G_{p-2,0,1}\left( \frac{\pi ^2}{y} \right)\nonumber\\
			&\quad-(-1)^{(p-1)/2}\frac{\pi^{p+1}}{y^{p+1}}\tilde G_{p-1,-1,2}\left( \frac{\pi ^2}{y} \right).
		\end{align}
		At this point, we shall compute the series $G_{p-1,0,1}\left( y \right)$. From Lemma \ref{cd Maclaurin expansion} and \cite{DCLMRT1992}, we acquire the Maclaurin expansion for $\mathrm{nd}(u)$
		\begin{align}\label{nd-exped-one}
			\mathrm{nd}(u)=\sum_{n=0}^{\infty}{\mathrm{A}_{2n}\left( x \right) \frac{\left( -1 \right) ^nu^{2n}}{\left(2n\right)!}}\quad \text{and}\quad\frac{2q^n}{1+q^{2n}}=\frac{1}{\cosh(ny)},
		\end{align}
        and its Fourier series representation
		\begin{align}
			&{\rm nd}(u)=\frac{\pi}{2Kk'}\frac{2\pi}{Kk'}\sum_{n=1}^{\infty}\frac{(-1)^n q^n}{1-q^{2n+1}}\cos\left[(2n)\frac{\pi u}{2K}\right]\quad (|q|<1).
		\end{align}
		Setting $q\equiv q\left( x \right):=e^{-y}$, we expand
		\begin{align}\label{nd-one}
			{\rm nd}(u)&=\frac{\pi}{2Kk'}+\frac{2\pi}{{K}k'}\sum_{n=1}^{\infty}{\frac{(-1)^n q^n n ^{2j}}{1+q^{2n}}}\sum_{j=0}^{\infty}{\frac{\left( -1 \right) ^j\left( \frac{\pi u}{{K}} \right) ^{2j}}{\left( 2j \right) !}}
			\nonumber\\
			&=\frac{\pi}{2{K}k'}+\frac{\pi}{{K}k'}\sum_{j=0}^{\infty}{\frac{\left( -1 \right) ^j\left( \frac{\pi u}{{K}} \right) ^{2j}}{\left( 2j \right) !}}G_{2j,0,1}\left( y \right).
		\end{align}	
		Through a comparison of the coefficients of $u^{2n}$ in \eqref{nd-exped-one} and \eqref{nd-one}, we derive
		\begin{align}\label{Bp,1}
			G_{p-1,0,1}\left( y \right)=(-1)^{(p-1)/2} z^p\sqrt{1-x}\frac{\mathrm{A}_{p-1}\left( x \right)}{2^p}.
		\end{align}
		Employing Lemma \ref{lem-2,transform}, for $\Omega \left( x,e^{-y},z,z' \right) =0$, we obtain
		\begin{align}\label{Bp,1-tansform}
			G_{p-1,0,1}\left(\frac{\pi^2}{y}\right)=(-1)^{(p-1)/2}\left(\frac{yz}{\pi}\right)^p\sqrt{x}\frac{\mathrm{A}_{p-1}\left(1-x\right)}{2^p}.
		\end{align}
		By using the results in Rui-Xu-Zhao's paper\cite[Thm. 3.2]{RXZ2023}, we have
		\begin{align}\label{Rui-two}
			\tilde G_{p,0,1}(y)=\sum_{n=1}^{\infty}\frac{(-1)^n (2n-1)^p}{\cosh\left(\frac{(2n-1)y}{2}\right)}
			=(-1)^{(p-1)/2-1}z^{p+1}\sqrt{x(1-x)}\frac{\mathrm{P}_{p}(x)}{2}.
		\end{align}

Consider the series $\tilde{G}_{p,0,1}(y)$,
and fix $\delta>0$. For $y\ge \delta$,
\[
\left|\frac{(-1)^n(2n-1)^p}{\cosh\!\left(\frac{2n-1}{2}y\right)}\right|
\le
\frac{(2n-1)^p}{\cosh\!\left(\frac{2n-1}{2}\delta\right)}
\le
2(2n-1)^p e^{-\tfrac{2n-1}{2}\delta},
\]
hence $\widetilde G_{p,0,1}(y)$ converges uniformly on $[\delta,\infty)$ by the Weierstrass M-test.
Moreover, differentiating
\[
f_n(y):=\frac{(-1)^n(2n-1)^p}{\cosh\!\left(\frac{2n-1}{2}y\right)}
\quad\text{gives}\quad
f_n'(y)= -\frac{(-1)^n(2n-1)^{p+1}}{2}\,
\frac{\sinh\!\left(\frac{2n-1}{2}y\right)}{\cosh^2\!\left(\frac{2n-1}{2}y\right)}.
\]
Using $\bigl|\sinh t\bigr|/\cosh^2 t \le 2e^{-t}$ for $t\ge0$, we obtain for $y\ge\delta$,
\[
|f_n'(y)|
\le C(2n-1)^{p+1}e^{-\tfrac{2n-1}{2}\delta},
\]
so $\sum_{n\ge1} f_n'(y)$ converges uniformly.

		Differentiating $ \tilde G_{p-2,0,1}\left( y \right) $ yields
		\begin{align*}
			\frac{\mathrm{d}}{\mathrm{d}y}\tilde G_{p-2,0,1}\left( y \right) =\sum_{n=1}^{\infty}{\frac{\mathrm{d}}{\mathrm{d}y}\frac{(-1)^n\left( 2n-1 \right) ^{p-2}}{\cosh \left( \frac{2n-1}{2}y \right)}=-\frac{1}{2}}\tilde G_{p-1,-1,2}\left( y \right).
		\end{align*}
Using the identity  $\frac{\mathrm{d}x}{\mathrm{d}y}=-x(1-x)z^2$ (see \cite[P. 120, Entry. 9(i)]{B1991} or \eqref{dx/dy}), we obtain
		\begin{align}\label{JS-relation-two}
			\tilde G_{p-1,-1,2}\left( y \right)=2x\left( 1-x \right) z^2\frac{\mathrm{d}}{\mathrm{d}x}\tilde G_{p-2,0,1}\left( y \right).
		\end{align}
		Applying Lemma \ref{lem-2,transform}, for $\Omega \left( x,e^{-y},z,z' \right) =0$, we deduce
		$$\Omega \left( 1-x,e^{-\pi ^2/y},yz/\pi ,\frac{1}{\pi}\left( \frac{1}{x\left( 1-x \right) z}-yz' \right) \right) =0.$$
		Consequently,  \eqref{Rui-two} and \eqref{JS-relation-two} can be rephrased as
		\begin{align}\label{Xp-2,1'}
			\tilde G_{p-2,0,1}\left( \frac{\pi ^2}{y} \right) =(-1)^{(p-1)/2}\left( \frac{yz}{\pi} \right) ^{p+1}\frac{\mathrm{P}_{p-2}\left( 1-x \right)}{2}\sqrt{x(1-x)}
		\end{align}
		and
		\begin{align}\label{DXp-2,2'}
			\tilde G_{p-1,-1,2}\left( \frac{\pi ^2}{y} \right)&=2x\left( 1-x \right) \left( \frac{yz}{\pi} \right) ^2\frac{\mathrm{d}}{\mathrm{d}\left( 1-x \right)}\tilde G_{p-2,0,1}\left( \frac{\pi ^2}{y} \right)\nonumber\\
			&=-2x\left( 1-x \right) \left( \frac{yz}{\pi} \right) ^2\frac{\mathrm{d}}{\mathrm{d}x}\tilde G_{p-2,0,1}\left( \frac{\pi ^2}{y} \right).
		\end{align}
		Finally, substituting \eqref{Bp,1-tansform}, \eqref{Xp-2,1'} and \eqref{DXp-2,2'} into \eqref{thm5.3-proof-begin-1} completes the proof through direct calculation.
	\end{proof}
	
	\begin{exa}
By evaluating specific instances of the right-hand side of \eqref{barS2p3} using \emph{Mathematica}, we obtain		
		\begin{align*}
			&\sum_{n = 1}^{\infty}\frac{(-1)^n (2n-1)^{2}}{\sinh((2n-1)y/2)\cosh^2((2n-1)y/2)}\\
&=\frac{1}{2}z^3\sqrt{(1-x) x}\Big\{4z'(1-x) x+(1-2 x) z\Big.-\sqrt{(1-x)}\Big\},\\
        &\sum_{n = 1}^{\infty}\frac{(-1)^n (2n-1)^{6}}{\sinh((2n-1)y/2)\cosh^2((2n-1)y/2)}\\
        &=\frac{1}{2}z^7\sqrt{x}\Big\{12z'\left(16 x^2-16x+1\right) (1-x)^{3/2}x\Big.+\left(-96 x^3+144 x^2-50x+1\right)\sqrt{1-x}z\Big.\\&\Big.\quad+61 x^3-107x^2+47 x-1\Big\},\\
        &\sum_{n = 1}^{\infty}\frac{(-1)^n (2n-1)^{10}}{\sinh((2n-1)y/2)\cosh^2((2n-1)y/2)}\\
			&=\frac{1}{2}z^{11}\sqrt{x}\Big\{20z'\left(7936 x^4-15872 x^3+9168 x^2-1232 x+1\right)(1-x)^{3/2}x\Big.\\&\Big.\quad+\left(-79360 x^5+198400 x^4-166112 x^3+50768 x^2-3698 x+1\right) \sqrt{1-x}z\Big.\\&\Big.\quad+50521 x^5-138933 x^4+130250 x^3-45530 x^2+3693 x-1 \Big\},\\
&\sum_{n = 1}^{\infty}\frac{(-1)^n (2n-1)^{14}}{\sinh((2n-1)y/2)\cosh^2((2n-1)y/2)}\\
			&=\frac{1}{2}z^{15} \sqrt{x}\Big\{28z'\left(22368256 x^6-67104768 x^5+71997696 x^4\atop-32154112 x^3+4992576 x^2-99648 x+1\right)(1-x )^{3/2}x\Big.\\&\Big.\quad+\left(-313155584 x^7+1096044544 x^6-1458129408 x^5\atop+905212160 x^4-255034240 x^3+25361472 x^2-298946 x+1\right) \sqrt{(1-x) }z\Big.\\&\Big.\quad+(199360981x^7-747603679 x^6+1074680289 x^5-728130163 x^4\atop+226132303 x^3-24738669 x^2+298939 x-1)\Big\}.
		\end{align*}
	\end{exa}

Setting \( p = 4m - 1 \) and \( x = \tfrac12 \) in Theorem~\ref{barS2p,2}, and noting that \( \mathrm{P}'_{4m-3}(1/2) = 0 \) (see the proof of Theorem 3.7 in \cite{RXZ2023}), we obtain the following corollary.
	\begin{cor}
		Set $\Gamma=\Gamma(1/4)$. For any integer $m>0$ we know that
		\begin{align}\label{js2-pi}
			\sum_{n = 1}^{\infty}\frac{(-1)^n (2n-1)^{4m-2}}{\sinh((2n-1)\pi/2)\cosh^2((2n-1)\pi/2)}&=-\frac{1}{2^{4m}}\mathrm{A}_{4m-2}\left(\frac{1}{2}\right)\frac{1}{\sqrt{2}}\frac{\Gamma ^{8m-2}}{\pi^{(12m-3)/2}}\nonumber\\
			&\quad+\frac{1}{2^{4m}}(4m-2)\mathrm{P}_{4m-3}\left(\frac{1}{2}\right)\frac{\Gamma ^{8m-4}}{\pi^{6m-2}}.
		\end{align}
	\end{cor}

	\begin{exa}
		Set $\Gamma=\Gamma(1/4)$. By letting $y=\pi$ and $x=1/2$ in Example $5.6$ with the help of \emph{Mathematica}, we arrive at the example below
		\begin{align*}
			\sum_{n = 1}^{\infty}\frac{(-1)^n (2n-1)^{2}}{\sinh((2n-1)\pi/2)\cosh^2((2n-1)\pi/2)}&=\frac{\Gamma^4}{8 \pi ^4}-\frac{\Gamma ^6}{32 \sqrt{2} \pi ^{9/2}},\\
			\sum_{n = 1}^{\infty}\frac{(-1)^n (2n-1)^{6}}{\sinh((2n-1)\pi/2)\cosh^2((2n-1)\pi/2)}&=\frac{27 \Gamma^{14}}{2048 \sqrt{2} \pi ^{21/2}}-\frac{9 \Gamma^{12}}{128 \pi ^{10}},\\
			\sum_{n = 1}^{\infty}\frac{(-1)^n (2n-1)^{10}}{\sinh((2n-1)\pi/2)\cosh^2((2n-1)\pi/2)}&=\frac{945 \Gamma^{20}}{2048 \pi ^{16}}-\frac{11529 \Gamma^{22}}{131072 \sqrt{2} \pi ^{33/2}},\\
			\sum_{n = 1}^{\infty}\frac{(-1)^n (2n-1)^{14}}{\sinh((2n-1)\pi/2)\cosh^2((2n-1)\pi/2)}&=\frac{23444883 \Gamma^{30}}{8388608 \sqrt{2} \pi ^{45/2}}-\frac{480249 \Gamma^{28}}{32768 \pi ^{22}}.
		\end{align*}
	\end{exa}

	\begin{thm}{\label{barGp,1}}
		Let $p\ge 3$ be an odd integer. We obtain
		\begin{align}\label{barGp,1-relation}
			\tilde{G}_{p,2,1}\left( y \right)& =\frac{1}{2}z^{p+1}\mathrm{P}_{p}(1-x)\sqrt{x(1-x)}-pz^{p+1}z'x(1-x)\mathrm{A}_{p-1}(1-x)\sqrt{x}\\\nonumber
			&\quad-x(1-x)z^{p+2}\frac{\mathrm{d}}{\mathrm{d}x}\left[\mathrm{A}_{p-1}(1-x)\sqrt{x}\right],
		\end{align}
		where $\mathrm{P}_{p}\left(x\right)$ and $\mathrm{A}_{p-1}\left(x\right)$ are the coefficient polynomials from the Maclaurin expansions of $\mathrm{sd}\left( u \right)$ and $\mathrm{nd}\left( u \right)$ respectively.
	\end{thm}
	\begin{proof}
		We initiate our derivation with equation (\ref{js3}), which gives
		\begin{align}\label{thm5.5-proof-begin-1}
			\tilde{G}_{p,2,1}\left( y \right) &=-(-1)^{(p-1)/2}\frac{\pi^{p+1}}{y^{p+1}}\tilde{G}_{p,0,1}\left( \frac{\pi ^2}{y} \right) -(-1)^{(p-1)/2}\frac{p2^p\pi^p}{y^p}G_{p-1,0,1}\left( \frac{\pi ^2}{y} \right)\nonumber\\
			&\quad+(-1)^{(p-1)/2}\frac{2^p\pi^{p+2}}{y^{p+2}}G_{p,-1,2}\left( \frac{\pi^2}{y} \right).
		\end{align}
		The derivative of $ G_{p-1,0,1}\left( y \right) $ satisfies
		\begin{align*}
			\frac{\mathrm{d}}{\mathrm{d}y}G_{p-1,0,1}\left( y \right) =\sum_{n=1}^{\infty}\frac{\mathrm{d}}{\mathrm{d}y}\frac{(-1)^n n^{p-1}}{\cosh(ny) }=-G_{p,-1,2}\left( y \right).
		\end{align*}
		From the fundamental relation $\frac{\mathrm{d}x}{\mathrm{d}y}=-x(1-x)z^2$ (\cite[P. 120, Entry. 9(i)]{B1991} or \cite{RXZ2023}), we have
		\begin{align}\label{JS-relation-three}
			G_{p,-1,2}\left( y \right)=x\left( 1-x \right) z^2\frac{\mathrm{d}}{\mathrm{d}x}G_{p-1,0,1}\left( y \right).
		\end{align}
		Utilizing Lemma \ref{lem-2,transform}, for $\Omega \left( x,e^{-y},z,z' \right) =0$, we obtain
		$$\Omega \left( 1-x,e^{-\pi ^2/y},yz/\pi ,\frac{1}{\pi}\left( \frac{1}{x\left( 1-x \right) z}-yz' \right) \right) =0.$$
		In addition, it is obvious to obtain from \eqref{Bp,1-tansform} and \eqref{Xp-2,1'} that
		\begin{align}\label{Xp,1'}
			\tilde{G}_{p,0,1}\left( \frac{\pi ^2}{y} \right) =-(-1)^{(p-1)/2}\left( \frac{yz}{\pi} \right) ^{p+1}\frac{\mathrm{P}_{p}\left( 1-x \right)}{2}\sqrt{x(1-x)}
		\end{align}
		and
		\begin{align}\label{DBp,2}
			G_{p,-1,2}\left( \frac{\pi ^2}{y} \right)&=x\left( 1-x \right) \left( \frac{yz}{\pi} \right) ^2\frac{\mathrm{d}}{\mathrm{d}\left( 1-x \right)}G_{p-1,0,1}\left( \frac{\pi ^2}{y} \right)\nonumber\\
			&=-x\left( 1-x \right) \left( \frac{yz}{\pi} \right) ^2\frac{\mathrm{d}}{\mathrm{d}x}G_{p-1,0,1}\left( \frac{\pi ^2}{y} \right).
		\end{align}
		Finally, substituting \eqref{Bp,1-tansform}, \eqref{Xp,1'} and \eqref{DBp,2} into \eqref{thm5.5-proof-begin-1} and performing the necessary calculations completes the proof.
	\end{proof}
	
	\begin{exa}
By evaluating the right-hand side of \eqref{barGp,1-relation} using \emph{Mathematica}, we present several specific examples as follows:
		\begin{align*}
		&\sum_{n=1}^{\infty}{\frac{(-1)^n(2n-1)^3}{\sinh ^2((2n-1)y/2)\cosh\mathrm{((}2n-1)y/2)}}
		\\
		&=\frac{1}{2}\sqrt{x}z^4\{-6z'(x-1)^2x-(x-1)(3x-1)z.+(1-2x)\sqrt{1-x}\},
		\\
		&\sum_{n=1}^{\infty}{\frac{(-1)^n(2n-1)^7}{\sinh ^2((2n-1)y/2)\cosh\mathrm{((}2n-1)y/2)}}
		\\
		&=\frac{1}{2}\sqrt{x}z^8\left\{ \begin{array}{c}
			\left( -272x^3+408x^2-138x+1 \right) \sqrt{1-x}\\
			-14z'(x-1)\left( 61x^3-107x^2+47x-1 \right) x\\
			-(x-1)\left( 427x^3-535x^2+141x-1 \right) z\\
		\end{array} \right\} ,
		\\
		&\sum_{n=1}^{\infty}{\frac{(-1)^n(2n-1)^{11}}{\sinh ^2((2n-1)y/2)\cosh\mathrm{((}2n-1)y/2)}}
		\\
		&=\frac{1}{2}\sqrt{x}z^{12}\left\{ \begin{array}{c}
			\left( -353792x^5+884480x^4-729728x^3+210112x^2-11074x+1 \right) \sqrt{1-x}\\
			-(x-1)\left( 555731x^5-1250397x^4+911750x^3-227650x^2+11079x-1 \right) z\\
		-22z'(x-1)\left( 50521x^5-138933x^4+130250x^3-45530x^2+3693x-1 \right) x\\
		\end{array} \right\} ,
		\\
		&\sum_{n=1}^{\infty}{\frac{(-1)^n(2n-1)^{15}}{\sinh ^2((2n-1)y/2)\cosh\mathrm{((}2n-1)y/2)}}
		\\
		&
		=\frac{1}{2}\sqrt{x}z^{16}\left\{ \begin{array}{c}
			\left( \begin{array}{c}
				-1903757312x^7+6663150592x^6-8804878080x^5+5354318720x^4\\
				-1429612624x^3+121675512x^2-896810x+1\\
			\end{array} \right) \sqrt{1-x}\\
			-(x-1)\left( \begin{array}{c}
				2990414715x^7-9718847827x^6+11821483179x^5\\
				-6553171467x^4+1582926121x^3-123693345x^2+896817x-1\\
			\end{array} \right) z\\
			-30z'(x-1)\left( \begin{array}{c}
				199360981x^7-747603679x^6+1074680289x^5-728130163x^4\\
				+226132303x^3-24738669x^2+298939x-1\\
			\end{array} \right) x\\
		\end{array} \right\} .
		\end{align*}
	\end{exa}

Setting \( p = 4m - 1 \) and \( x = 1/2 \) in Theorem \ref{barGp,1}, and noting that \( \mathrm{P}_{4m-1}(1/2) = 0 \) (see the proof of Theorem 3.7 in \cite{RXZ2023}), we derive the following corollary.
	\begin{cor}
		Set $\Gamma=\Gamma(1/4)$ and $\mathrm{A}_p'(x):=\mathrm{d}\mathrm{A}_p(x)/\mathrm{d}x$. For any integer $m>0$, we have
		\begin{align}\label{js3-pi}
	&\sum_{n = 1}^{\infty}\frac{(-1)^n(2n-1)^{4m-1}}{\sinh^2((2n-1)\pi/2)\cosh((2n-1)\pi/2)}\nonumber\\
    &=-\frac{1}{2^{4m}}(4m-1)\mathrm{A}_{4m-2}\left(\frac{1}{2}\right)
			\frac{1}{\sqrt{2}}\frac{\Gamma ^{8m-2}}{\pi^{(12m-1)/2}}\nonumber\\
	&\quad-\frac{1}{2^{4m+3}}\left(\mathrm{A}'_{4m-2}\left(\frac{1}{2}\right)+\mathrm{A}_{4m-2}\left(\frac{1}{2}\right)\right)\frac{1}{\sqrt{2}}\frac{\Gamma ^{8m+2}}{\pi^{(12m+3)/2}}.
		\end{align}
	\end{cor}

	\begin{exa}
		Set $\Gamma=\Gamma(1/4)$. By \emph{Mathematica}, we substitute $y=\pi $ and $ x=1/2$ into Example $5.10$ and obtain the following identities
		\begin{align*}
			\sum_{n = 1}^{\infty}\frac{(-1)^n(2n-1)^{3}}{\sinh^2((2n-1)\pi/2)\cosh((2n-1)\pi/2)}&=\frac{\Gamma^{10}}{256 \sqrt{2} \pi ^{15/2}}-\frac{3 \Gamma^6}{32 \sqrt{2} \pi ^{11/2}},\\
			\sum_{n = 1}^{\infty}\frac{(-1)^n(2n-1)^{7}}{\sinh^2((2n-1)\pi/2)\cosh((2n-1)\pi/2)}&=\frac{189 \Gamma^{14}}{2048 \sqrt{2} \pi ^{23/2}}-\frac{87 \Gamma^{18}}{16384 \sqrt{2} \pi ^{27/2}},\\
			\sum_{n = 1}^{\infty}\frac{(-1)^n(2n-1)^{11}}{\sinh^2((2n-1)\pi/2)\cosh((2n-1)\pi/2)}&=\frac{57969 \Gamma^{26}}{1048576 \sqrt{2} \pi ^{39/2}}-\frac{126819 \Gamma^{22}}{131072 \sqrt{2} \pi ^{35/2}},\\
			\sum_{n = 1}^{\infty}\frac{(-1)^n(2n-1)^{15}}{\sinh^2((2n-1)\pi/2)\cosh((2n-1)\pi/2)}&=\frac{351673245 \Gamma^{30}}{8388608 \sqrt{2} \pi ^{47/2}}-\frac{160692903 \Gamma^{34}}{67108864 \sqrt{2} \pi ^{51/2}}.
		\end{align*}
	\end{exa}

	\section{Evaluations of Berndt-type Integrals}
	\begin{thm}\label{mixInt1}
		Let $m\in \mathbb{N}$ and $\Gamma=\Gamma(1/4)$. Then, the following integral evaluates as
		\begin{align}
			\int_0^{\infty}{\frac{x^{4m-1}\mathrm{d}x}{\left[ \cosh(2x)-\cos(2x) \right] \left[ \cosh x+\cos x \right]}}&= c_{1,m} \frac{\Gamma ^{8m-4}}{\pi ^{2m-1}}+ \frac{c_{2,m}}{\sqrt{2}}\frac{\Gamma ^{8m-2}}{\pi^{(4m-1)/2}}\\\nonumber
			&\quad+c_{3,m} \frac{\Gamma ^{8m}}{\pi^{2m}}+ \frac{c_{4,m}}{\sqrt{2}}\frac{\Gamma ^{8m+2}}{\pi^{(4m+3)/2}}+c_{5,m} \frac{\Gamma ^{8m+4}}{\pi^{2m+3}},
		\end{align}
		where the constants $c_{1,m},c_{2,m},c_{3,m},c_{4,m},c_{5,m}\in \mathbb{Q}$.
	\end{thm}
	\begin{proof}
		Following the work of Rui-Xu-Zhao's\cite{RXZ2023}, let $\Gamma=\Gamma(1/4)$, and for positive integers $n$ and $k$, define $ \mathrm{P}_n^{(k)}:=\mathrm{P}_n^{(k)}(1/2) $, $\mathrm{S}_n^{(k)}:=\mathrm{S}_n^{(k)}(1/2)$ and $\mathrm{A}_n^{(k)}=\mathrm{A}_n^{(k)}(1/2)$.

		Besides, $ \mathrm{P}_{4m-1}(1/2)=\mathrm{P}_{4m-3}'(1/2)=0$. Then, for positive integer  $m>0$ we know that
		\begin{align}\label{Rui-last result}
			\sum_{n=1}^{\infty}\frac{(-1)^n(2n-1)^{4m-1}}{\cosh^3\left(\frac{(2n-1)\pi}{2}\right)}=\frac{\Gamma^{8m-4}}{2^{4m+7}\pi^{6m + 3}}
			\left\{
			\begin{array}{l}
				128(8m^2-6m+1)\pi^4 p_{4m-3}\\
				+\Gamma^8[(4m-6)p_{4m-3}+p_{4m-3}'']
			\end{array}
			\right\}.
		\end{align}
		By substituting the equations \eqref{js1-pi}, \eqref{js2-pi}, \eqref{js3-pi} and \eqref{Rui-last result} into Theorem \ref{thm1cosh++}, we obtain
		\begin{align}
			&2\int_0^{\infty}{\frac{x^{4m-1}\mathrm{d}x}{\left[ \cosh(2x)-\cos(2x) \right] \left[ \cosh x+\cos x \right]}}\nonumber\\
&=(-1)^{m-1}\frac{1}{2^{6m+1}}(8m^2-6m+1)\mathrm{P}_{4m-3}\frac{\Gamma^{8m-4}}{\pi^{2m-1}}\nonumber\\
			&\quad+(-1)^{m-1}\frac{1}{2^{6m+2}}(4m-1)(\mathrm{S}_{4m-2}-\mathrm{A}_{4m-2})\frac{1}{\sqrt{2}}\frac{\Gamma^{8m-2}}{\pi^{(4m-1)/2}}\nonumber\\
			&\quad+(-1)^{m-1}\frac{1}{2^{8m+3}}q_{4m-2}(-1)\frac{\Gamma^{8m}}{\pi^{2m}}\nonumber\\
			&\quad+(-1)^{m-1}\frac{1}{2^{6m+5}}(\mathrm{S}_{4m-2}'-\mathrm{S}_{4m-2}+\mathrm{A}_{4m-2}'+\mathrm{A}_{4m-2})\frac{1}{\sqrt{2}}\frac{\Gamma^{8m-2}}{\pi^{(4m+3)/2}}\nonumber\\
			&\quad+(-1)^{m}\frac{1}{2^{6m+8}}\left[(4m-6)\mathrm{P}_{4m-3}+\mathrm{P}_{4m-3}''\right]\frac{\Gamma^{8m+4}}{\pi^{2m+3}}.
		\end{align}
		From this, the coefficients $c_{i,m}$ can be explicitly identified as
		\begin{align}
			c_{1,m}&=(-1)^{m-1}\frac{1}{2^{6m+2}}(8m^2-6m+1)\mathrm{P}_{4m-3}\in \mathbb{Q},\label{mainresult-proofcoffic1}\\
			c_{2,m}&=(-1)^{m-1}\frac{1}{2^{6m+3}}(4m-1)(\mathrm{S}_{4m-2}-\mathrm{A}_{4m-2})\in \mathbb{Q},\label{mainresult-proofcoffic2}\\
			c_{3,m}&=(-1)^{m-1}\frac{1}{2^{8m+4}}q_{4m-2}(-1)\in \mathbb{Q},\label{mainresult-proofcoffic3}\\
			c_{4,m}&=(-1)^{m-1}\frac{1}{2^{6m+6}}(\mathrm{S}_{4m-2}'-\mathrm{S}_{4m-2}+\mathrm{A}_{4m-2}'+\mathrm{A}_{4m-2})\in \mathbb{Q},\label{mainresult-proofcoffic4}\\
			c_{5,m}&=(-1)^{m}\frac{1}{2^{6m+9}}\left[(4m-6)\mathrm{P}_{4m-3}+\mathrm{P}_{4m-3}''\right]\in \mathbb{Q}.\label{mainresult-proofcoffic5}
		\end{align}
		Thus, we have completed the proof of this theorem.
	\end{proof}

\begin{re}
Note that $K({1}/{\sqrt{2}}) = \frac{\Gamma^2(1/4)}{4\sqrt{\pi}}$ (see \cite{BZ1992}). Consequently, the right-hand side of \eqref{mixInt1} can also be expressed in an alternative form:
\begin{align}
		\int_0^{\infty}{\frac{x^{4m-1}\mathrm{d}x}{\left[ \cosh(2x)-\cos(2x) \right] \left[ \cosh x+\cos x \right]}}&\in \mathbb{Q}\left[4K({1}/{\sqrt{2}})\right]^{4m-2}+\frac{\mathbb{Q}}{\sqrt{2}}\left[4K({1}/{\sqrt{2}})\right]^{4m-1}\nonumber\\
		&\quad+\mathbb{Q}\left[4K({1}/{\sqrt{2}})\right]^{4m}+ \frac{\mathbb{Q}}{\sqrt{2}}\frac{\left[4K({1}/{\sqrt{2}})\right]^{4m+1}}{\pi}\nonumber\\&\quad+\mathbb{Q} \frac{\left[4K({1}/{\sqrt{2}})\right]^{4m+2}}{\pi^{2}}.
	\end{align}
The Gamma function form is preferred for our focus on numerical computation, special value correlations, and multiple zeta functions, while the $K$-function form is more appropriate for elliptic integral theory, result generality, and field-specific conventions. These expressions are provided for readers' reference.
\end{re}
Using \eqref{mainresult-proofcoffic1}-\eqref{mainresult-proofcoffic5} and assisted by \emph{Mathematica}, we compute the values of the five constant coefficients appearing on the right-hand side of the formula in Theorem \ref{mixInt1} for $m=1, 2, 3,4$. This yields the following examples.
	\begin{exa}\label{ExampleC}
		Let $\Gamma=\Gamma(1/4)$. We have
		\begin{align*}
			\int_0^{\infty}{\frac{x^{3}\mathrm{d}x}{\left[ \cosh(2x)-\cos(2x) \right] \left[ \cosh x+\cos x \right]}}&=\frac{\Gamma^ {12}}{16384\pi^5}-\frac{\Gamma^{10}}{4096 \sqrt{2}\pi^{7/2}}+\frac{\Gamma^8}{2048\pi^2}\\
			&\quad-\frac{3\Gamma^6}{512 \sqrt{2}\pi^{3/2}}+\frac{3\Gamma ^4}{256 \pi},\\
			\int_0^{\infty}{\frac{x^{7}\mathrm{d}x}{\left[ \cosh(2x)-\cos(2x) \right] \left[ \cosh x+\cos x \right]}}&=\frac{13 \Gamma^{20}}{1048576 \pi ^7}-\frac{87 \Gamma^{18}}{1048576 \sqrt{2}\pi ^{11/2}}+\frac{9 \Gamma^{16}}{65536 \pi ^4}\\
			&\quad-\frac{189 \Gamma ^{14}}{131072 \sqrt{2} \pi ^{7/2}}+\frac{63 \Gamma^{11}}{16384 \pi^3},\\
			\int_0^{\infty}{\frac{x^{11}\mathrm{d}x}{\left[ \cosh(2x)-\cos(2x) \right] \left[ \cosh x+\cos x \right]}}&=\frac{2169 \Gamma ^{28}}{67108864 \pi ^9}-\frac{57969 \Gamma^{26}}{268435456 \sqrt{2} \pi ^{15/2}}+\frac{189 \Gamma^{24}}{524288 \pi ^6}\\
			&\quad-\frac{126819 \Gamma ^{22}}{33554432  \sqrt{2} \pi ^{11/2}}+\frac{10395 \Gamma^{20}}{1048576\pi ^5},\\
			\int_0^{\infty}{\frac{x^{15}\mathrm{d}x}{\left[ \cosh(2x)-\cos(2x) \right] \left[ \cosh x+\cos x \right]}}&=\frac{1504197 \Gamma ^{36}}{4294967296 \pi^{11}}-\frac{160692903 \Gamma^{34}}{68719476736  \sqrt{2}\pi^{19/2}}\\
			&\quad+\frac{130977 \Gamma^{32}}{33554432  \pi^8}-\frac{351673245 \Gamma^{30}}{8589934592 \sqrt{2} \pi ^{15/2}}\\&\quad+\frac{7203735 \Gamma^{28}}{67108864 \pi^7}.
		\end{align*}

	\end{exa}

	\section{Berndt-type Integrals via Barnes Multiple Zeta Functions}\label{sec-Berndt-Barmzf}
In this section, we evaluate certain Barnes multiple zeta values via Berndt-type integrals.
Building upon the work of Bradshaw and Vignat \cite{BV2024}, who expressed generalized
Berndt-type integrals \eqref{BTI-definition-1} in terms of the (alternating) Barnes multiple
zeta function, we record the following key identities. For instance, they proved
\cite[Prop.~2]{BV2024} that
\begin{align}
	&\int_0^\infty \frac{x^{a}\mathrm{d}x}{(\cos x-\cosh x)^b}
	=2^b \Gamma(a+1)\ze_{2b}(a+1,b|(1+i,1-i)^b)\quad (a\geq 2b,\ b\geq 1),\label{BMZF-RTI-}\\
	&\int_0^\infty \frac{x^{a}\mathrm{d}x}{(\cos x+\cosh x)^b}
	=2^b \Gamma(a+1){\bar \ze}_{2b}(a+1,b|(1+i,1-i)^b)\quad (a\geq 0,\ b\geq 1)\label{BMZF-RTI+}.
\end{align}
Here $\bfs^n$ denotes the $n$-fold repetition of the string $\bfs$. Combined with the
explicit evaluations of Xu and Zhao \cite{XZ2024}, these formulas yield closed-form
expressions for the corresponding Barnes multiple zeta values (see \cite[Cor.~2]{BV2024}).

For $a_1,\ldots,a_N>0$ (write ${\bf a}_N:=\{a_1,\ldots,a_N\}$), the \emph{Barnes multiple zeta
	function} and its \emph{alternating} analogue are defined by
\cite{AIK2014,Ba1904,BV2024,KMT2023} as
\begin{align}
	\ze_N(s,\omega|{\bf a}_N)
	:=\sum_{n_1\geq0,\ldots,n_N\geq 0} \frac{1}{(\omega+n_1a_1+\cdots+n_Na_N)^s}
	\quad (\Re(\omega)>0,\ \Re(s)>N),\label{defn:Barneszetafunction}
\end{align}
and
\begin{align}
	{\bar \ze}_N(s,\omega|{\bf a}_N)
	:=\sum_{n_1\geq0,\ldots,n_N\geq 0} \frac{(-1)^{n_1+\cdots+n_N}}{(\omega+n_1a_1+\cdots+n_Na_N)^s}
	\quad (\Re(\omega)>0,\ \Re(s)> N-1),\label{defn:ABarneszetafunction}
\end{align}
respectively.

Bradshaw and Vignat \cite{BV2024} also derived a convenient Laplace--Mellin formula for these
series, recovering a known identity \cite[Eq.~(3.2)]{R2000}.
\begin{pro}
	Let $\Re(s)>N$, $\Re(\omega)>0$, and $\Re(a_j)>0$ for $j = 1,\ldots,N$. Then
\begin{align*}
&\ze_N(s,\omega|a_1,\ldots,a_N)
	=\frac{1}{\Gamma(s)}\int_{0}^{\infty}u^{s - 1}e^{-\omega u}\prod_{j = 1}^{N}(1 - e^{-a_ju})^{-1}\,\mathrm{d}u,\\
&\bar{\ze}_N(s,\omega|a_1,\ldots,a_N)
	=\frac{1}{\Gamma(s)}\int_{0}^{\infty}u^{s - 1}e^{-\omega u}\prod_{j = 1}^{N}(1 + e^{-a_ju})^{-1}\,\mathrm{d}u.
\end{align*}
\end{pro}

In a recent paper \cite{Zhou2025}, the third author of this paper introduced the following \emph{generalized Barnes multiple zeta function}:
\begin{align}
	\ze _N(s,\omega |\mathbf{a}_N;\boldsymbol{\sigma }_N)
	:=\sum_{n_1\ge 0,\ldots,n_N\ge 0}
	\frac{\sigma_1^{n_1}\cdots \sigma_N^{n_N}}{(\omega +n_1a_1+\cdots +n_Na_N)^s}
	\quad (\Re(\omega)>0,\ \Re(s)>N),
\end{align}
where $\boldsymbol{\sigma }_N=(\sigma _1,\ldots ,\sigma _N)\in\{\pm 1\}^N$. Moreover, Zhou \cite[Prop. 5.9]{Zhou2025} also established the following integral representation:
\begin{pro}
	Let $\Re(s)>N$, $\Re(\omega)>0$, and $\Re(a_j)>0$ for $j = 1,\ldots,N$. Then
	\[
	\ze _N(s,\omega |\mathbf{a}_N;\boldsymbol{\sigma }_N)
	=\frac{1}{\Gamma(s)}\int_0^{\infty}u^{s-1}e^{-\omega u}\prod_{j=1}^N\bigl(1-\sigma _j e^{-a_j u}\bigr)^{-1}\,\mathrm{d}u.
	\]
\end{pro}
	
	\begin{thm}\label{GBZeta} For positive integer $m$, we have
		\begin{align*}
			\Gamma \left( 4m \right)\zeta _4\left( 4m,3|\mathbf{c}_4;\boldsymbol{\sigma }_4 \right)\in \mathbb{Q}\frac{\Gamma ^{8m-4}}{\pi ^{2m-1}}+ \frac{\mathbb{Q}}{\sqrt{2}}\frac{\Gamma ^{8m-2}}{\pi^{(4m-1)/2}}
			+\mathbb{Q}\frac{\Gamma ^{8m}}{\pi^{2m}}+\frac{\mathbb{Q}}{\sqrt{2}}\frac{\Gamma ^{8m+2}}{\pi^{(4m+3)/2}}+\mathbb{Q}\frac{\Gamma ^{8m+4}}{\pi^{2m+3}},
		\end{align*}
		where $\mathbf{c}_4=\left( 2+2i,2-2i,1+i,1-i \right) $ and $\boldsymbol{\sigma }_4=\left( \left\{ 1 \right\}^2,\left\{ -1 \right\}^2 \right)$.
	\end{thm}
	\begin{proof}
		Following \cite[Prop. 2]{BV2024}, we obtain
		\begin{align*}
			\int_0^{\infty}{\frac{x^{s-1}e^{-\omega x}}{\prod_{i=1}^M{\sinh \left( a_ix \right)}\prod_{j=1}^N{\cosh \left( b_jx \right)}}}\mathrm{d}x&=\sum_{\tiny\begin{array}{c}
					n_1,\ldots ,n_M\ge 0\\
					k_1,\ldots ,k_N\ge 0\\
			\end{array}}^{}{\frac{2^{M+N}\Gamma \left( s \right) \left( -1 \right) ^{k_1+\cdots +k_N}}{\left( w+\sum_{i=1}^M{a_i+\sum_{i=1}^N{b_i+}2\left( \sum_{i=1}^M{a_in_i}+\sum_{i=1}^N{b_ik_i} \right)} \right) ^s}}\\
			&=2^{M+N}\Gamma \left( s \right) \zeta _{M+N}\left( s,w+\sum_{i=1}^M{a_i+}\sum_{i=1}^N{b_i}|\boldsymbol{c}_{M+N};\boldsymbol{\sigma }_{M+N} \right),
		\end{align*}
		where $
		\boldsymbol{c}_{M+N}=\left( 2a_1,\ldots,2a_M,2b_1,\ldots ,2b_N \right)$ and $\boldsymbol{\sigma }_{M+N}=\left( \left\{ 1 \right\}^M,\left\{ -1 \right\}^N \right).$

By applying the product-to-sum formulas for trigonometric (and hyperbolic) functions, direct calculations yield
		\begin{align*}
			&\int_0^{\infty}{\frac{x^{4m-1}\mathrm{d}x}{\left[ \cosh(2x)-\cos(2x) \right] \left[ \cosh x+\cos x \right]}}=4\Gamma \left( 4m \right)\zeta _4\left( 4m,3|\mathbf{c}_4,\boldsymbol{\sigma }_4 \right)\\
			&=\frac{1}{4}\int_0^{\infty}{\frac{x^{4m-1}\mathrm{d}x}{\sinh \left[ \left( 1+i \right) x \right] \sinh \left[ \left( 1-i \right) x \right] \cosh \left( \frac{1+i}{2}x \right) \cosh \left( \frac{1-i}{2}x \right)}},
		\end{align*}
		where $\mathbf{c}_4=\left( 2+2i,2-2i,1+i,1-i \right) $ and $\boldsymbol{\sigma }_4=\left( \left\{ 1 \right\}^2,\left\{ -1 \right\}^2 \right)$.

		Finally, using Theorem \ref{mixInt1} yields the desired evaluation.
	\end{proof}
Employing the computational framework of \emph{Mathematica} in conjunction with Theorem \ref{GBZeta}, we perform explicit calculations for $m=1,2,3,4$.
	\begin{exa}
		Let $\Gamma=\Gamma(1/4)$. Then
		\begin{align*}
			\zeta _4\left( 4,3|\mathbf{c}_4,\boldsymbol{\sigma }_4 \right)&=\frac{\Gamma^{12}}{393216 \pi ^5}-\frac{\Gamma^{10}}{98304 \sqrt{2} \pi ^{7/2}}+\frac{\Gamma^8}{49152 \pi ^2}-\frac{\Gamma^6}{4096 \sqrt{2} \pi ^{3/2}}+\frac{\Gamma^4}{2048 \pi },\\
			\zeta _4\left( 8,3|\mathbf{c}_4,\boldsymbol{\sigma }_4 \right)&=\frac{13 \Gamma^{20}}{21139292160 \pi ^7}-\frac{29 \Gamma ^{18}}{7046430720 \sqrt{2} \pi ^{11/2}}+\frac{\Gamma^{16}}{146800640 \pi ^4}-\frac{3 \Gamma^{14}}{41943040 \sqrt{2} \pi ^{7/2}}\\&\quad+\frac{\Gamma ^{12}}{5242880 \pi ^3},\\
        \zeta _4\left( 12,3|\mathbf{c}_4,\boldsymbol{\sigma }_4 \right)&=\frac{241 \Gamma^{28}}{1190564934451200 \pi ^9}-\frac{2147 \Gamma ^{26}}{1587419912601600 \sqrt{2} \pi ^{15/2}}+\frac{\Gamma^{24}}{442918502400 \pi ^6}\\&\quad-\frac{61 \Gamma^{22}}{2576980377600 \sqrt{2} \pi ^{11/2}}+\frac{\Gamma^{20}}{16106127360 \pi ^5},\\
        \zeta _4\left( 16,3|\mathbf{c}_4,\boldsymbol{\sigma }_4 \right)&=\frac{55711 \Gamma^{36}}{832062021389254656000 \pi ^{11}}-\frac{40487 \Gamma^{34}}{90564573756653568000 \sqrt{2} \pi ^{19/2}}\\&\quad+\frac{\Gamma^{32}}{1340029796352000 \pi ^8}-\frac{179 \Gamma^{30}}{22869841857740800 \sqrt{2} \pi ^{15/2}}+\frac{11 \Gamma^{28}}{536011918540800 \pi ^7},
		\end{align*}
		where $\mathbf{c}_4=\left( 2+2i,2-2i,1+i,1-i \right) $ and $\boldsymbol{\sigma }_4=\left( \left\{ 1 \right\}^2,\left\{ -1 \right\}^2 \right)$.
	\end{exa}

	{\bf Conflict of Interest.} The authors declare no conflict of interest regarding the publication of this article.
	
	{\bf Data Availability} No new data were generated or analyzed in this study, and therefore data sharing is not applicable.
	
	{\bf Use of AI Tools Declaration.} The authors confirm that no artificial intelligence (AI) tools were used in the creation of this work.

	{\bf Acknowledgments.}
The authors would like to thank the anonymous referees for their valuable comments and suggestions. Ce Xu is supported by the General Program of Natural Science Foundation of Anhui Province (Grant No. 2508085MA014).

\end{document}